\documentclass[11pt]{article}

\usepackage{fullpage}

\usepackage{amssymb,hyperref,amsthm,mathtools,enumerate,thmtools,thm-restate}
\usepackage[capitalise,nameinlink, noabbrev]{cleveref}

\usepackage[vlined,boxed,commentsnumbered, ruled,linesnumbered]{algorithm2e}

\SetKwInOut{Given}{Input}
\SetKwFor{RepTimes}{repeat}{times}{}

\usepackage{dsfont}

\usepackage[noend]{algpseudocode}

\usepackage[margin=1in]{geometry}
\newcommand{\mc}{\mathcal}

\declaretheorem[name=Theorem, parent=section]{theorem}

\declaretheorem[name=Lemma, sibling=theorem]{lemma}
\declaretheorem[name=Claim, sibling=theorem]{claim}

\newtheorem*{remark}{Remark}
\declaretheorem[name=Fact, sibling=theorem]{fact}
\theoremstyle{definition}
\declaretheorem[name=Definition, sibling=theorem]{definition}

\theoremstyle{remark}

\newcommand\E{\mathbb{E}}

\newcommand\N{\mathbb{N}}

\newcommand\F{\mathbb{F}}

\newcommand\card[1]{\left| {#1} \right|}

\newcommand\set[1]{{\left\{ #1 \right\}}}

\newcommand\Prob[2]{{\Pr_{#1}\left[ {#2} \right]}}

\newcommand\Expc[2]{{\mathop{\mathbb{E}}_{#1}\left[ {#2} \right]}}

\newcommand\KMconst{c_{_{\text{KM}}}}

\newcommand\ceil[1]{\left\lceil{#1}\right\rceil}

\newcommand\floor[1]{\left\lfloor{#1}\right\rfloor}

\newcommand\eps{\varepsilon}

\newcommand\ind[1]{{\mathds 1}_{#1}}

\newcommand{\RM}{{\sf RM}}

\newcommand{\oracle}[1]{\mathcal{O}_{#1}}

\newcommand{\sem}{{\sf semi}}
\newcommand{\ka}{k_{\sf adv}}
\newcommand{\dist}{{\sf dist}}
\newcommand{\Dist}{{\text{Dist}}}

\newcommand{\qbin}[2]{\begin{bmatrix}{#1}\\ {#2}\end{bmatrix}_q}
\DeclareMathOperator{\var}{var}
\usepackage{soul}
\usepackage[dvipsnames]{xcolor}

\newcommand{\tester}[1]{{\mathrm{SSB}_{{#1}}}}

\title{Optimal Testing of Reed-Muller Codes with an Online Adversary}

\author{}
 \author{
 Esty Kelman\thanks{Department of Mathematics, Massachusetts Institute of Technology and Boston University, Cambridge, MA, USA. Supported in part by the National Science Foundation under Grant No. 2022446 and the NSF TRIPODS program (award DMS-2022448). Email: \texttt{ekelman@mit.edu}} 
 \and
Uri Meir\thanks{Blavatnik School of Computer Science, Tel Aviv University, Israel. Email: \texttt{urimeir.cs@gmail.com}}
 \and
 Kai Zhe Zheng\thanks{Department of Mathematics, Massachusetts Institute of Technology, Cambridge, MA, USA. Supported by the NSF GRFP DGE-2141064. Email: \texttt{kzzheng@mit.edu}}
 }

\date{}

\begin{document}

\maketitle

\begin{abstract} 

Motivated by applications to property testing in the online-erasure model of Kalemaj, Raskhodnikova, and Varma (ITCS 2022 and Theory of Computing 2023), we define and analyze {\em semi-sample-based testers} for Reed-Muller codes. The task in Reed-Muller testing is to determine whether an input function $f: \F^n \to \F$ belongs to the Reed-Muller code or is far from it, using as few point queries to $f$ as possible. Reed-Muller testing is a well-studied task with its roots in both the Property Testing and Probabilistically Checkable Proofs literature.
The online-erasure model introduces a twist: after each query made, an adversary may erase up to $t$ points of the input function, potentially thwarting any test in which the queries follow a predictable pattern.

Semi-sample-based testers are a hybrid between sample-based testers --- which can only make uniformly random queries to the input function --- and standard testers, which can choose their queries freely. They are designed with the online-erasure model in mind and operate by first choosing some subset $S$ of the domain and then making their queries uniformly at random inside of $S$. 
We describe semi-sample-based testers for the Reed-Muller code and give an optimal analysis of their soundness. 

Consequently, we show that semi-sample-based testers are indeed effective in the presence of online erasures, and thereby achieve optimal query complexity for testing the Reed-Muller code in the online-erasure model. This result improves upon prior work of Minzer and Zheng (SODA 2024). As an added bonus, we show that semi-sample-based testers also exist for the lifted affine-invariant codes of Guo, Kopparty, and Sudan (ITCS 2013), thereby providing the first known testers for these codes in the online-erasure model.

\end{abstract}

\newpage
\section{Introduction}
Local testing of the Reed-Muller code is a fundamental problem in property testing and complexity theory. In this problem, we are given oracle access to a multivariate function over a finite field, $f : \F^n \to \F$, and must decide whether $f$ is a polynomial of total degree at most $d$, or is instead far from every polynomial of degree at most $d$. The notion of distance used here is the fractional Hamming distance. The goal is to design a randomized algorithm, called a \emph{tester}, that accepts $f$ with probability $1$ in the former case and rejects $f$ with probability at least $2/3$ in the latter case, while making as few queries to $f$ as possible. Ideally, the query complexity should be independent of $n$.

The problem of Reed-Muller testing has been extensively explored over the past few decades. Its study began with the degree-$1$ case, which was considered in the seminal work of Blum, Luby, and Rubinfeld~\cite{BLR93}. Their work introduced the well-known BLR linearity test and, more broadly, initiated the field of property testing. Subsequently, a large body of work gave and optimized testers for the general degree $d$ case, $d > 1$~\cite{AlonKKLR05,KaufmanR06,BKSSZ10,HaramatySS13,RZS,kaufman2025improved,MinZCodes}. Additionally, a related variant of the Reed-Muller testing problem was instrumental in the early construction of Probabilistically Checkable Proofs (PCPs)~\cite{AS92,RazS97}, and has continued to inspire subsequent research, including~\cite{MoshkovitzR08,BDN,MZcube,HKSS,BS}.

More recently, a new model of property testing introduced by Kalemaj, Raskhodnikova, and Varma, called the \emph{online-erasure model}~\cite{KalemajRV22}, has garnered significant attention. Several works have revisited Reed-Muller testing and related questions in this setting~\cite{KalemajRV22,MinZ,ben2024property,WYZ,AKM,KLR,KelmanMNDR25}. In this model, an adversary may \emph{erase} up to $t$ inputs after each query $f(x)$ is made, where typically $t$ is assumed to be small relative to the domain size $|\F|^n$. Testers that work in the online-erasure model are called \emph{online-erasure-resilient}.

We continue this line of work and give improved online-erasure-resilient testers for the Reed-Muller code. Our testers have query complexity that improves upon the previous state of the art in~\cite{MinZ}. In addition, we give a matching lower bound showing that, for fixed field size and degree parameter, our query complexity is optimal as a function of the erasure rate parameter $t$.

Toward this result, we find it convenient to define a new class of testers, which we refer to as \emph{semi-sample-based testers}. At a high level, semi-sample-based testers make their queries uniformly at random within some moderately large set of points $S$. In the case of Reed-Muller testing, the set $S \subset \F^n$ is chosen to have size independent of $n$, but still sufficiently large as a function of $t$, $d$, and $q := |\F|$. By making $S$ large enough relative to the number of queries made, we ensure that our semi-sample-based testers are also online-erasure-resilient. Beyond Reed-Muller codes, we note that our approach extends naturally to give semi-sample-based testers for the lifted affine-invariant codes of~\cite{GKS}. To the best of our knowledge, these are the first testers for these codes that are resilient to online erasures.

We believe that the notion of semi-sample-based testing is conceptually useful and may find further applications in property testing, particularly in adversarial models such as the online-erasure model. To motivate this notion and explain why it arises naturally in this setting, we examine several natural candidate testers for Reed-Muller codes and show why they fail under online erasures. These candidates represent two extremes, which semi-sample-based testers naturally interpolate between. For simplicity, we restrict the discussion below to the case $d = 1$, as these issues already arise in this setting.

\paragraph{The Standard Tester.}
A first approach one might try is to run the standard tester in the online-erasure model. This tester is the well-known BLR linearity tester, which chooses $x, y \in \F^n$ uniformly at random, queries $f(x)$, $f(y)$, and $f(x+y)$, and then checks whether $f(x) + f(y) = f(x+y)$. While this tester has optimal query complexity in the standard model, it fails completely in the online-erasure model. Indeed, an adversary can erase the value $f(x+y)$ after the first two queries are made, making it impossible to carry out the test. The issue here is that the queries are \emph{too structured}.

\paragraph{Sample-based testing.}
At the other extreme, one can design a tester whose queries are as unstructured as possible. For example, one may consider a tester that queries $f$ at uniformly random points $x \in \F^n$. Such testers were introduced in~\cite{GGR98} and are called \emph{sample-based testers}. In this case, it is unlikely that an erased point is ever queried (since $t \ll |\F|^n$), but this robustness comes at a cost: the tester requires at least $n$ queries to succeed. Indeed, as shown in \cite{GoldreichR2016}, with high probability, any set of fewer than $n$ uniformly random points is linearly independent, and therefore there always exists a linear function that agrees with $f$ on those points. Thus, while sample-based testers are robust to online erasures, they are not query efficient.

\paragraph{Semi-sample-based testing.}
Semi-sample-based testers interpolate between these two extremes. They operate by first choosing a subset of the domain and then taking uniformly random samples from within that subset. Our tester, in particular, chooses a $k$-dimensional linear subspace $U$ and queries $f$ at uniformly random points from $U$. Note that the case $k = n$ corresponds to sample-based testing, while the case $k = 1$ closely resembles the BLR tester.

By choosing the dimension of $U$ appropriately, we obtain the best of both worlds. We achieve erasure resilience by taking $k$ sufficiently large relative to $t$, while retaining query efficiency by keeping $k$ independent of $n$.

\subsection{Main Results}

Our main result is the construction and analysis of a new tester for the Reed-Muller code using the semi-sampled-based approach. Along the way, we note that our semi-sample-based testers can also be made to work for other affine-invariant families, such as lifted affine-invariant codes \cite{GKS}. Throughout this section, we say that a function $f: \F^n \to \F$ is $\eps$-far from a given family of functions if, for every function in that family, $f$ differs from it on at least an $\eps$-fraction of all possible inputs.

\paragraph{Reed-Muller Codes.}
Fix a degree parameter $d \in \N$ and a finite field $\F$ of size $q$, where $q$ is a power of a prime $p$. The starting point for our tester is the simple observation that the restriction of any polynomial of degree at most $d$ to a linear subspace remains a polynomial of degree at most $d$. 

Formally, let $U \subset \F^n$ be a linear subspace of dimension $k$ with basis $v_1, \ldots, v_k$. If $f : \F^n \to \F$ has degree at most $d$, then the induced $k$-variate function $f|_U$, defined by
\[
f|_U(z_1, \ldots, z_k) := f\left(\sum_{i=1}^k z_i v_i\right),
\]
is a polynomial of degree at most $d$ in the variables $z_1, \ldots, z_k$. This observation underlies most previous Reed-Muller testers, including those of~\cite{AlonKKLR05,BKSSZ10,HaramatySS13,kaufman2025improved}.

In the standard model, the tester proceeds by choosing a subspace $U$ of dimension $k = k_{q,d} + 1$, where
\[
k_{q,d} := \ceil{\frac{d+1}{q - q/p}},
\]
and checking whether $f|_U$ has degree at most $d$ by querying $f(x)$ for every point $x \in U$. Our semi-sample-based variant, which is instead parameterized by $k \geq k_{q,d} + 1$ and denoted by $\tester{k}$, queries $f$ on a random subset of $s_k$ points from $U$, rather than on all of $U$. The precise choice of $s_k$ is stated in \Cref{th: semi-sample rm}, but crucially it satisfies
\[
    s_k = O\!\left(q^{k_{q,d}+1} \left(\frac{e(d+k)}{k}\right)^k \right) .
\]

\paragraph{The tester $\tester{k}$.} 
The degree-$d$ tester with parameter $k$ proceeds  as follows:
\begin{enumerate}
    \item Choose a linear subspace $U \subset \F^n$ of dimension $k$ uniformly at random.

    \item Choose $s_k$ points from $U$ uniformly at random, and let $S$ denote the resulting set.
      \item Accept if there exists a degree $d$ function that agrees with $f$ on every point in $S$. Otherwise, reject.
\end{enumerate}

We note $\tester{k}$ makes at least as many queries as the $q^{k_{q,d}+1}$ made by the standard tester, and its query complexity increases with $k$.
However, $\tester{k}$ has two properties which are important for us: (1) all queries made are uniformly random inside $U$; and (2) the \emph{fraction} of points queried in the chosen subspace $U$ equals to $s_k/q^k$ which goes to $0$ as $k$ increases.
These properties turn out to be crucial for our desired application to testing in the online erasure model.

The focus of this paper is thus to analyze $\tester{k}$ above, and our main result is the following.
\begin{theorem}[Informal Version of \cref{th: semi-sample rm}] \label{thm:main1}
The tester $\tester{k}$ for degree $d$ Reed--Muller codes over $\F^n$ satisfies the following:
\begin{itemize}
    \item \emph{Completeness:} If $\deg(f) \leq d$, then the tester accepts with probability $1$.
    \item \emph{Soundness:} If $f$ is $\eps$-far from every degree $d$ function, then the tester rejects with probability at least
    \[
    \min\{1/128,\; s_k \eps / 8\}.
    \]
\end{itemize}
\end{theorem}

To obtain a tester for Reed-Muller codes in the online erasure model, we set $k$ sufficiently large with respect to the erasure parameter $t$, momentarily setting $\eps = \Omega(1)$.
Concretely, we choose $k = \Theta(d) + \log_q t$ and observe that for $k$ at least this large, the probability of querying an erased point is negligible. This observation leads to the following result.

\begin{theorem}[Informal Version of \cref{thm: rm online erasure test}] \label{thm: main erasure}
For any degree parameter $d$, field $\F$ of size $q$, dimension $n$, and erasure parameter
$t \leq O(q^{\,n - \Theta(d)})$, there exists a tester for the degree $d$ Reed-Muller code over $\F^n$ in the $t$-online-erasure model with query complexity $\left(O_q\!\left(\frac{\log t}{d}\right)\right)^d$  that satisfies the following:
\begin{itemize}
    \item \emph{Completeness:} If $\deg(f) \leq d$, then the tester accepts with probability $1$.
    \item \emph{Soundness:} If $f$ is $\Omega(1)$-far from every degree $d$ function, then the tester rejects with probability at least $\Omega(1)$.
\end{itemize}
\end{theorem}

Prior to our work, the best-known online-erasure-resilient Reed-Muller testers used $O_q(\log t)^{3d}$ queries over prime fields and $O_q(\log t)^{3d+3q}$ queries over non-prime fields. Our results improve upon these bounds in all cases. We complement this upper bound with an almost matching lower bound, showing that our query complexity is optimal (up to constant factors) as a function of the erasure parameter for fixed field size and degree.

\begin{theorem}[Informal Version of \cref{thm:lower bound for all fields}] \label{thm: erasure lower bound inf}
For any degree parameter $d$, field $\F$ of size $q$, dimension $n$, and erasure parameter $t$, any tester for the degree $d$ Reed-Muller code over $\F^n$ must make at least $
\left(\frac{\log_q t}{d}\right)^d$ queries.
\end{theorem}

Previously, this lower bound was shown over $\F_2$ in~\cite{ben2024property}; we note that, by leveraging recent results from \cite{beameOY2018,golovnevGHNY2024}, the argument extends naturally to general finite fields.

\paragraph{Lifted Affine Invariant Codes}
In addition to testing the Reed-Muller code, we note that our semi-sample-based testers can also be defined quite naturally for a family of functions called lifted affine-invariant codes \cite{GKS}. These codes are a generalization of Reed-Muller codes and were previously studied in the context of property testing in \cite{GKS, HRZS, kaufman2025improved}. They are defined as follows. Let $\mc{C}_j \subseteq \{\F^j \to \F\}$ be an \emph{affine-invariant} family of functions, meaning $f \in \mc{C}_j$ if and only if $f \circ T \in \mc{C}_j$ for any affine transformation $T: \F^j \to \F^j$. Then, the $n$-dimensional lift of $\mc{C}_j$ is $\mc{C}^{j \nearrow n}$ and is the set of functions $f: \F^n \to \F$ such that $f|_U \in \mc{C}_j$ for every $j$-dimensional affine subspace $U \subset \F^n$. 

It is straightforward to design a semi-sample-based tester for $\mc{C}^{j \nearrow n}$. Indeed, one may set any $k\geq j$ and adjust the number of random queries in step 2 suitably (denoted by $Q_k$ below). Then, in step 3 the tester checks if there is a function from $\mc{C}^{j \nearrow k}$ agreeing with the queries. Our main result for lifted affine-invariant codes is the following.

\begin{theorem}[Informal Version of \cref{th: semi-sample lifted}]  \label{thm: main 1}
    Let $\mc{C}_n = \mc{C}^{j \nearrow n}$ be the $n$-dimensional lift of the affine invariant code $\mc{C}_j \subseteq \{\F^j \to \F\}$ and let $\delta_0$ be the smallest minimum distance of the codes $\mc{C}_j, \cdots, \mc{C}_n$. Choose dimension $k\geq j$ and set the number of queries to be
    \[
    Q_k := \ceil{\frac{100\ln |\mc{C}_k|}{\delta_{0}}}+ 1.
    \]
    The semi-sample-based tester for $\mc{C}^{j \nearrow n}$ satisfies the following:
    \begin{itemize}
        \item Completeness: If $f \in \mc{C}^{j \nearrow n}$, then the tester accepts with probability $1$.
        \item Soundness: If $f$ is $\eps$-far from $\mc{C}^{j \nearrow n}$, then the tester rejects with probability at least 
        \[
        \min\{Q_k \eps/8, 1/128 \}.
        \]
    \end{itemize}
\end{theorem}

Once again, by choosing $k$ suitably, we also get the following online-erasure-resilient tester for lifted affine-invariant codes. We state it informally below, saving the precise upper threshold for the erasure parameter to be formally defined later.
\begin{theorem}[Informal Version of \cref{thm: lifted t online erasure tester}] \label{thm: lifted t online erasure tester_informal}
Let $\mc{C}^{j \nearrow n}$ be a lifted affine-invariant code and let the erasure parameter $t$ and dimension parameter $k$ be chosen suitably. Then $\mc{C}^{j \nearrow n}$ has an $O( Q_k)$ query tester in the $t$-online erasure model that satisfies the following
  \begin{itemize}
    \item Completeness: If $f \in \mc{C}_n$, then the tester accepts with probability $1$.
    \item Soundness: If $f$ is $\Omega(1)$-far from $\mc{C}^{j \nearrow n}$, the tester rejects with probability at least $2/3$.
    \end{itemize}
\end{theorem}

\paragraph{Optimal Soundness Testing and Comparison with \cite{MinZ}.} \label{par:optimal testing}

In \cite{MinZ}, a similar tester is analyzed in the online-erasure model and a version of \cref{thm: main erasure} is obtained with a higher query complexity\footnote{The result of~\cite{MinZ} states $3d+3$ in the exponent, but their analysis easily shows an exponent of $3d$.} of $O_q(\log(t))^{3d}$. We discuss the differences between our result and the main result of \cite{MinZ}. The tester in \cite{MinZ} can be viewed as $\tester{k}$ with a modified third step. Instead of checking if there is a degree $d$-function agreeing with $f$ on all queried points, the tester of \cite{MinZ} finds a structured dual codeword, over domain $U$, that is supported on the queried points, and checks if $f|_U$ is indeed dual to this codeword. That is, \cite{MinZ} finds a function $h: U \to \F_q$ such that every $f: U \to \F_q$ with $\deg(f) \leq d$ has inner product $0$ with $h$ (i.e. $\sum_{x \in U} f(x) h(x) = 0$) and checks whether this is the case. Note that computing this inner product only requires knowing $f(x)$ for $x$ in the support of $h$, hence the reason for requiring $h$ to be supported on the queried points. From here, \cite{MinZ} rely on the generic soundness result of \cite{KS08}, which 
guarantees, under certain conditions satisfied here, that a $Q$-query tester rejects constant-far functions with probability roughly $Q^{-2}$.
Thus, while their base tester makes $O_q(\log(t))^{d}$ queries, the soundness they obtain for this test requires $O_q(\log(t))^{2d}$ further repetitions to reject with constant probability, and hence leads to a version of \cref{thm: main erasure} with $O_q(\log(t))^{3d}$.

By adapting the third step of $\tester{k}$ to accept if and only if there is a degree $d$ function explaining the queried points, we are able to obtain (in \cref{thm:main1}) a better soundness analysis for $\tester{k}$ which rejects constant far functions with constant probability (and hence leads to \cref{thm: main erasure} with the stated query complexity). In fact, our soundness analysis for the $\tester{k}$ is \emph{optimal} up to constant factors. As described in \cite{kaufman2025improved}, a tester making $Q$ queries is called an \emph{optimal tester} if it achieves rejection probability $\min\set{\Omega(1), \Omega(Q \eps)}$ in the soundness case. Such testers are optimal in the sense that this is the best probability with which one can hope to reject an $\eps$-far function with $Q$ queries. Indeed, thinking of an $\eps$-far function as a valid function with $\eps$-fraction of its entries perturbed, the probability of even querying a perturbed point, and hence having any hope of rejecting, is at most $\min\set{\Theta(1), Q \eps}$.

Showing that a fixed tester has such optimal soundness is typically a difficult task. For example, the first optimal Reed-Muller testers were only shown relatively recently: in \cite{BKSSZ10} when $|\F| = 2$, and in \cite{HaramatySS13, kaufman2025improved, MinZCodes} when $|\F| > 2$. The works mentioned only show that the specific tester they consider is optimal, and do not easily extend to other testers such as our semi-sample-based ones or the general testers in \cite{KS08}.

Finally, an additional factor of $(\Theta(d))^d$ in the query complexity comes from a precise choice of $k = \log_q t + \Theta(d)$ (rather than $k = \Theta(d \log_q t)$). This results in the main term of $(\Theta(\frac{\log t}{d}))^d$ in the query complexity, matching with our lower bound in \Cref{thm: erasure lower bound inf}.

\subsection{Technical Overview}
\label{sec:intro-techniques}

We describe our techniques towards our main theorem for semi-sample-based Reed-Muller testers \cref{thm:main1}. The approach for lifted affine-invariant codes is similar. Our analysis follows the inductive approach pioneered by \cite{BKSSZ10}. However, for our purposes, we need to make a few key enhancements, which use techniques from the PCP literature.
In particular, we show a hyperplane agreement lemma (in the low error regime), that is motivated by the Plane versus Point analysis of \cite{RazS97}, as well as the high-error hyperplane agreement lemmas that were key to recent breakthroughs on 2-to-2 games \cite{KhotMS18} and optimal alphabet-soundness tradeoff PCPs \cite{MZpcp}. Similar ideas were also recently used in the testing of direct-sums \cite{WYZ}.

\paragraph{Proof outline.}
Let $s$ be the number of queries made, and consider $f$ of distance exactly $\eps$ from having degree $d$. Let $f = g + h$ be such that $g$ is the degree $d$ function that is closest to $f$ and $h$ is a function with fractional support size $\eps$.
As in \cite{BKSSZ10}, we split the soundness analysis into a small distance case, say where $\eps < \frac{1}{10s}$, and a large distance case, where $\eps \geq \frac{1}{10s}$. In the small distance case, one can show directly that with probability roughly $s \cdot\eps$, the tester will only query one input from the support of $h$ and in this case, one can easily argue that the tester will indeed reject.

More difficult to analyze is the large distance case, where one relies on induction on the ambient dimension $n$ and assumes good soundness is guaranteed for all smaller ambient dimensions.
Here we want to show that the test rejects with some constant probability, say $1/100$.
The key piece of the inductive step is a lemma of the following form:
 \begin{lemma} \label{lm: hyp intro}
            If there are $M = q^{\Theta(1)} \cdot s$ hyperplanes $W_1,\dots, W_M$ on which the restriction of $f$ to any $W_i$ is 
        $\frac{1}{100s}$-close to degree $d$, then $f$ is $\frac{1}{10s}$-close to degree $d$.
\end{lemma}
Observe that we can perform the semi-sample-based tester of \cref{thm:main1} by first choosing a hyperplane $W$ and then performing the tester on the restriction of $f$ to $W$, denoted by $f|_W$, which is a function over a space of dimension $n-1$. Thus, we can write the rejection probability as $\Prob{}{\text{$f$ is rejected}} = \E_{W}\Pr[f|_W \text{ is rejected}]$. If there are more than $M$ hyperplanes on which $f|_W$ is $\frac{1}{100s}$-close to degree $d$, then we get that $f|_W$ is $\frac{1}{10s}$-close to degree $d$ and can apply the small distance analysis.
On the other hand, if there are fewer than $M$ such hyperplanes, then with probability roughly $1 - M/q^{n-1}$ over the choice of $W$, the function $f|_W$ is still relatively far from degree $d$, and by the inductive hypothesis we reject with some constant probability.
In order to ensure that we still get constant soundness overall though, we must ensure that $M/q^{n-1}$ is a small constant, and in particular, this means we must start our induction at dimension roughly $n = \log_q(M)$. Fortunately, \cref{lm: hyp intro} gives a strong enough bound on $M$ above for the Reed-Muller code and take $n$ roughly $\Theta(k)$ to get \cref{thm:main1}. 

We note that similar statements to \cref{lm: hyp intro} are shown in \cite{BKSSZ10, HaramatySS13}, but neither of these is sufficient for our purposes.
The reason is because \cite{BKSSZ10} only works over field size $q = 2$ while \cite{HaramatySS13} requires obtains $M = c(q) \cdot s$ for some $c(q)$ with tower type dependence on $q$ (rather than polynomial in $q$ as in \cref{lm: hyp intro}), due to the use of the density Hales-Jewett theorem. Additionally, these arguments do not apply beyond Reed-Muller codes, whereas we ultimately want our techniques to generalize to other affine-invariant codes and function families. 

Our proof of \cref{lm: hyp intro} instead employs techniques from agreement testing and PCPs. In \cite{RazS97}, for example, a similar statement is shown except with planes in place of hyperplanes and non-trivial agreement with a degree $d$ function on the planes, instead of being close to degree $d$ on the planes. Similarly \cite{KhotMS18, MZpcp} show a similar statement to \cref{lm: hyp intro}, but with non-trivial agreement on the hyperplanes. Thus, \cref{lm: hyp intro} can be viewed as a low-error version of the lemmas in \cite{RazS97, KhotMS18, MZpcp}, and at a high level, the proof proceeds in the same way.
We consider the \emph{consistency graph} over the $M$ hyperplanes, $W_1, \ldots, W_M$, and say that two vertices $i, j \in [M]$ are adjacent if the local degree $d$ functions on $W_i, W_j$, call them $f_i, f_j$, agree on the intersection $W_i \cap W_j$.
By showing that this graph is ``nearly  transitive'' we find a large clique, which corresponds to many of the $f_i$'s agreeing in pairs. From here we arbitrarily extrapolate the $f_i$'s from the large clique to some global function $F$ and show that 1) $F$ is close to $f$; and 2) $F$ is close to degree $d$. By the triangle inequality, this establishes \cref{lm: hyp intro}.

While \cref{lm: hyp intro} on its own right is not hard to establish with existing techniques, we believe its statement and the way we use to obtain testing results are generic enough and could find more applications, making it one novel aspect of our work. In particular, using the same outline we extend our results to families of affine-invariant properties.
The main bulk of our strategy follows through, save for step 2). 

In order to accomplish 2), we rely on known local testing results to show that $F$ passes a local test with high probability and is thus close to the desired family of functions. Indeed, if all the tester's queries fall within one of the $W_i$'s above in the clique, then the test surely passes, as here $F$ is consistent with the legitimate function $f_i$. If $M$ is sufficiently large, then this probability is high.
That is, step 2) relies on a known (typically structured) test, and leads the way to obtaining unstructured tests for the same property.

Currently, we require some lower bound on the soundness of known testers for the analysis of our tester to work (specifically in step 2). For future work, one could hope to perform step 2) algebraically for some fixed properties, as was done for Reed Muller codes in \cite{BKSSZ10, HaramatySS13}. For non-linear affine-invariant properties, such as some of those considered in \cite{BFL, bfhhl}, such an approach could yield significantly improved soundness.

\section{Preliminaries}

\paragraph{Notation.}
In general, we let $\F$ denote a finite field with characteristic $p$ and cardinality $\card{\F} = q = p^{\ell}$.
At times we will use a subscript such as $\F_q$ or $\F_2$ to specify a specific field size, but we omit the subscript when it is clear from context. For a function $f: \F^n \to \F$, it is well known that there is an $n$-variate polynomial with degree in each variable at most $|\F|-1$ whose evaluation over $\F^n$ is equal to $f$. We let $\deg(f)$ denote the total degree of this polynomial. Through the paper, for the notion of distance, we will use the fractional Hamming distance. For two functions $f, g:\F^n \to \F$, the fractional Hamming distance is
\[\dist(f,g) := \frac{ \card{\set{x\in\F^n: f(x) \neq g(x)}}}{q^n}.\]
We extend the definition to a family of functions $\mc{F}\subseteq \{\F^n\to\F\}$ by 
\[\dist(f,\mc{F}) := \min_{g\in\mc{F}} \dist(f,g)\ .\]
The function families that we consider typically form an error-correcting code, meaning any two members of the family are far apart in fractional hamming distance. For a function family $\mc{F}$, we define its distance as
\[
\delta(\mc{F}) = \min_{f, g \in \mc{F}, f \neq g} \dist(f, g).
\]

We let $\mc{G}_k$ be the set of linear subspaces of dimension $k$ inside of $\F^n$, where the ambient space $\F^n$ will always be clear from context. We denote the number of distinct linear subspaces of dimension $k$ inside a linear space of dimension $n$ by
\[
    \qbin{n}{k} := \card{\mc{G}_k} = \prod_{i=0}^{k-1}\ \frac{q^n-q^i}{q^k-q^i} .
\]
We will also work with affine subspaces of dimension $k$, which are sets of the form $x + U = \{x + y\; |\; y \in U \}$, where $U\in\mc{G}_k$. We use $\dim(U)$ to denote the dimension of $U$ when it is either an affine or linear subspace. It will always be clear from context which kind of subspace $U$ is.

\paragraph{Affine Invariance.}
A transformation $T: \F^n \to \F^n$ is called an affine transformation if it can be expressed as
\[
    T(x) = Mx + b,
\]
for some matrix $M \in \F^{n \times n}$ and affine shift $b \in \F^n$. 

A family of functions $\mc{C}_n \subseteq \{\F^n \to \F\}$ is called affine invariant if for any affine transformation, even non-invertible, $T: \F^n \to \F^n$ the following holds:
\[
f \in \mc{C}_n \implies f \circ T \in \mc{C}_n,
\]
where $f \circ T$ is the function given by $f \circ T (x) = f(T(x))$. Note that we require the above to hold even for $T$ non-invertible.

For a $k$-dimensional subspace $U$ (either linear or affine), we let $f|_U: \F^k \to \F$ denote the restriction of $f$ to $U$. For our purposes, we will not care about the parametrization of $U$, so we can define $f|_U = f \circ T$ by taking any affine transformation $T: \F^k \to \F^n$ whose image is $U$.
For every dimension $k$, let 
\[
\mc{C}_k = \{f|_U \; | \; f\in \mc{C}_n, U \in \mc{G}_k \}.
\]
One can check that if $\mc{C}_n$ is affine invariant, then so is $\mc{C}_k$. A common local test for $\mc{C}_n$, which we call the $k$-space test, goes as follows:
\begin{itemize}
    \item Choose $U \in \mc{G}_k$ uniformly at random.
    \item Accept if $f|_U \in \mc{C}_k$. Otherwise, reject.
\end{itemize}

Note that since $\mc{C}_k$ is affine invariant, it does not matter what parametrization we choose for $f|_U$ when performing the test. We remark that typically $\mc{C}_k$ is defined using affine-subspaces (or flats), instead of only linear subspaces as we do above, and the local test considered is usually based on restrictions to flats, rather than to linear subspaces. This test is typically called the $(k-1)$-flat test, but note that we can simulate a $(k-1)$-flat test using a $k$ dimensional linear subspace test. Indeed, instead of checking $f|_{U'}$ for a random $(k-1)$-flat $U'$, we can instead sample a random linear subspace $U$ of dimension $k$ containing $U'$. Over random $U'$, $U$ is uniformly random in $\mc{G}_k$, and hence performing the $k$-space test above with $f|_U$ is just as good as the $(k-1)$-flat test. 

\paragraph{The Reed-Muller Code.}
Given a finite field vector space $\F_q^n$, we refer to the set of degree at most $d$ polynomials as the degree $d$ Reed-Muller code and denote this by $\RM[n,q,d]$. 

We will need the following well-known facts about the rate, distance and local testability of Reed-Muller codes.
\begin{fact} \label{fact: rm distance}
    The code $\RM[n,q,d]$ has distance $q^{-\frac{d}{q-1}}$ and rate at most $\binom{n + d}{d}/q^n$. That is, for any two distinct functions $f, g \in \RM[n,q,d]$,
    \[
    \Pr_{x \in \F_q^n}[f(x) \neq g(x)] \geq q^{-\frac{d}{q-1}},
    \]
    and 
    \[
    \frac{\log_q |\RM[n,q,d]|}{q^n} \leq \frac{\binom{n + d}{d}}{q^n}.
    \]
    For the rate we are using the fact that $\RM[n,q,d]$ is spanned by all total degree at most $d$ monomials over $n$ variables and there are at most $\binom{n+d}{d}$ such monomials.
\end{fact}

Finally, given field size $q$ and degree $d$, we define the testing dimension of $\RM[n,q,d]$ to be 
\[
t_{q,d} = \ceil{\frac{d+1}{q-q/p}} + 1.
\]
We will use $t_{q,d}$ to refer to the above quantity throughout the paper. The following result from \cite{kaufman2025improved} shows the following soundness for the local test which restricts to a random $t_{q,d}-1$ affine subspace.

\begin{theorem}[\protect{\cite[Theorem 1.3]{kaufman2025improved}}] \label{thm: flat test rm}
    There exists a constant $\KMconst$ such that if $f$ is $\eps$-far from $\RM[n,q,d]$, then
\[
\Pr_{U' \subset \F_q^n}[\deg(f|_{U'}) > d] \geq q^{-\KMconst}\cdot\min\{1, q^{t_{q,d}}\eps\},
\]
where $U'$ is a random \emph{affine} subspace of dimension $t_{q,d}-1$.
\end{theorem}

As discussed above, we can also simulate the test above by restricting to a linear subspace of one higher dimension. Its soundness is an immediate consequence of \cref{thm: flat test rm}, and we give the short proof below.
\begin{theorem}\label{thm: rm local testability}
    For the same constant $\KMconst$ from \Cref{thm: flat test rm} it holds that if $f$ is $\eps$-far from $\RM[n,q,d]$, then
\[
    \Pr_{U \subset \F_q^n}[\deg(f|_{U}) > d] \geq q^{-\KMconst}\cdot\min\{1, q^{t_{q,d}}\eps\},
\]
where $U$ is a random \emph{linear} subspace of dimension $t_{q,d}$.
\end{theorem}
\begin{proof}
We can sample a random $t_{q,d}$-space by first choosing a $(t_{q,d}-1)$-flat $U'$, and then choosing a random linear subspace $U$ containing $U'$. Note that $U$ is a uniformly random $t_{q,d}$-space. It follows that,

   \[
    \Pr_{U \subset \F_q^n}[\deg(f|_{U}) > d] = \Expc{U' \subset \F_q^n}{\Pr_{U \supset U'}[\deg(f|_{U'}) > d]} = \Pr_{U' \subset \F_q^n}[\deg(f|_{U'}) > d] \geq q^{-\KMconst}\cdot\min\{1, q^{t_{q,d}}\eps\}.
   \]
\end{proof}

\paragraph{Lifted Affine-Invariant Codes.} Lifted affine-invariant codes were introduced by Guo, Kopparty, and Sudan \cite{GKS}. We formally define the codes below.

\begin{definition}[\protect{\cite[Definition 1.1]{GKS}}]
    Let $\mc{C} \subseteq \{\F^r \to \F \}$ be a linear affine-invariant family of functions. Then the lifted code $\mc{C}^{r \nearrow n}$ is the following set of functions:
    \[
    \{f: \F^n \to \F \; | \; f|_U \in \mc{C}_r, \forall U \text{ affine }, \dim(U) = r \}.
    \]
\end{definition}

When $\mc{C}$ in the above definition is clear, we will often use $\mc{C}_n$ to refer to the code $\mc{C}^{r \nearrow n}$. The code $\mc{C}$ is often referred to as the base code of $\mc{C}^{r \nearrow n}$ and $\mc{C}^{r \nearrow n}$ is referred to as the $n$-dimensional lift of $\mc{C}$. Lifted affine-invariant codes generalize Reed-Muller codes, indeed it can be seen that $\RM[n,q,d]$ is the $n$-dimensional lift of $\RM[t_{q,d}-1, q, d]$. In \cite{kaufman2025improved}, it was shown that a natural test based on restrictions to $r$-dimensional affine subspaces is a local test for $\mc{C}^{r \nearrow n}$ with optimal soundness. Their result also implies similar soundness for a test based on restricting to $(r+1)$-dimensional linear subspaces. We state this version of their result below as it is more convenient for us.

\begin{theorem}[\protect{\cite[Theorem 1.4]{kaufman2025improved}}] \label{thm: testing lifted km}
Given $f: \F^n \to \F$, the local test obtained by choosing a uniformly random linear subspace of dimension $r+1$ has the following guarantees:
\begin{itemize}
    \item Completeness: If $f \in \mc{C}_n$ then $\Pr_{U}[f|_U \in \mc{C}_{r+1}] =1$.
    \item Soundness: If $f$ is $\eps$-far from $\mc{C}_n$, then,
    \[
    \Pr_{U}[f|_U \notin \mc{C}_{r+1}] \geq q^{-O(1)}\min\left\{1, q^{r}\eps \right\}.
    \]
\end{itemize}
\end{theorem}

\paragraph{Distinguishing all members of a function family.}

\begin{claim}\label{claim:sample_based_ub general}
Fix $f: \F_q^n \to \F_q$ and let $\mathcal{F} \subseteq \{ \F_q^n \to \F_q \}$ be a family of functions with size $|\mathcal{F} |=N$ and distance $\delta(\mc{F}) = \delta$. Then choosing a set $S \subseteq \F_q^n$ of size
$M$ uniformly at random with replacement, with probability at least,
\[
1-(1-\delta)^M \cdot \binom{N}{2}\ge 1-e^{-\delta M}\cdot N^2.
\]
every pair of distinct functions in $\mc{F}$ disagree on at least one point in $\mc{F}$. As a consequence, there is at most one $F \in \mc{F}$ such that $f(x) = F(x), \forall x \in S$. 
\end{claim}
\begin{proof}
    Fix a pair of functions $F_1, F_2 \in \mc{F}$. Then the probability that they agree on all points in $S$ is at most  $(1-\delta)^M$.
    As $|\mc{F}| = N$, we can union bound over all pairs of functions in $\mc{F}$ to get that with probability at most 
$(1-\delta)^M \cdot \binom{N}{2}$, there exist two distinct functions in $\mc{F}$ that agree on $S$. The result follows.
\end{proof}

We conclude that whenever $M>>\frac{\log N}{\delta}$ the probability of $S$ not satisfying \cref{claim:sample_based_ub general} is negligible.
For example, let $\mc{F}\subset{\F_2^n\to\F_2}$ is the class of degree at most $d$ functions over $\F_2$ then, $\delta=\delta_d=2^{-d}$ and $N=2^{\binom{n}{\le d}}$ to have a good set with probability at least $0.9$ it is enough to take $M\approx 2^d \binom{n}{\le d}$ samples.

\section{The Hyperplane Agreement Lemma}
In this section, we show, for any affine-invariant family of functions $\mc{C}_n$, if $f$ is close to some function in $\mc{C}_n$ on sufficiently many hyperplanes, then $f$ must in fact be close to $\mc{C}_n$. The quality of our bound on how many hyperplanes are needed depends on the distance of the code as well as the local testability of $\mc{C}_n$.

Before stating the main lemma of the section, let us describe some necessary notation regarding the family $\mc{C}_n$. Recall that $\mc{G}_j$ denotes the set of $j$-dimensional subspaces of $\F_q^n$. Fix $\mc{C}_n \subseteq \{\F^n \to \F \}$ to be some affine-invariant family of functions and for each dimension $k$, let $\mc{C}_k \subseteq \{\F^n \to \F \}$ be the subset of functions given by, $\mc{C}_k = \{f|_U \; | \; U \subseteq \F^n, U \in \mc{G}_k, f \in \mc{C}_n\}$. Suppose that there is a dimension $t$ with the following properties:
\begin{definition} \label{def: t test assump}
    We say that a code $\mc{C}_n$ is locally testable via the $t$-space test with soundness function $c(\cdot)$, if for any function $f: \F^n \to \F$ it holds that
    \[
    \Pr_{U \in \mc{G}_t}[f|_U \notin \mc{C}_t] \geq  c(\dist(f, \mc{C}_n)),
    \]
    for some function $c: \mathbb{R} \to \mathbb{R}$ that is continuous, non-decreasing, and satisfies $c(0) = 0$.
\end{definition}
We fix this $t$ for the remainder of the section and let $\delta_0$ be the smallest minimum distance of the codes $\mc{C}_t,\mc{C}_{t+1}, \ldots, \mc{C}_n$. That is,
\begin{equation}\label{eq: min distance over lifted codes}
  \delta_0 := \min_{k: t\leq k \leq n} \delta(\mc{C}_k).
\end{equation}
The main lemma of this section states that if $f$ is far from $\mc{C}_n$, then $f|_W$ must also be far from $\mc{C}_{n-1}$ on almost all hyperplanes of $W \subset \F^n$. Formally,

\begin{lemma}[Hyperplane Agreement] \label{lm: hyperplane general}
Let $\mathcal{C}_n$ be an 
affine invariant family of functions over $\F_q^n$ that is locally testable via the $t$-space test with soundness function $c(\cdot)$ and let $\delta_0$ as in~\eqref{eq: min distance over lifted codes}. Set $\eps < \delta_0/6$ and let $M\geq \max\big\{\frac{10^5q^4}{\eps},\frac{20q^t}{c(\eps)}\big\}$. If there are $M$ distinct hyperplanes $W_1,\dots,W_M\subset \F^n$, such that for each $i=1,\dots, M$, the restriction of $f$ to the hyperplane $W_i$, denoted by $f|_{W_i}$, is $\eps$-close to $\mc{C}_{n-1}$ on $W_i$, then $f$ is $4\eps$-close to $\mc{C}_n$. 

\end{lemma}

In the lemma above, $M$ is set sufficiently large with respect to 
the soundness function, $c(\cdot)$, of the $t$-space test on $\mc{C}_n$, so that $\frac{20q^t}{M} \leq c(\eps)$. In other words, $M$ is sufficiently large so that if $\dist(f, \mc{C}_n) > \eps$, then the testing guarantee of \cref{def: t test assump} gives
    \[
    \Pr_{U \in \mc{G}_t}[f|_U \notin \mc{C}_t] \geq \frac{20q^t}{M}.
    \]

The remainder of the section is dedicated to showing \cref{lm: hyperplane general}. We follow the strategy of \cite{RazS97} in the plane versus point analysis. Suppose that there are many hyperplanes $W_i$ with local degree $d$ functions $f_i: W_i \to \F_q$ such that $f_i$ and $f$ are close. We will first show that many --- in fact at least a constant fraction --- of these $f_i$'s are actually consistent with each other. Using these consistent $f_i$'s we can then define a function $F$ on the whole space and argue that 1) $F$ is close to $f$ and 2) $F$ passes the $t_{q,d}$-space test with high probability. For both 1) and 2) we are crucially using the fact that $F$ is defined using many consistent local functions on hyperplanes, so for nearly every point $x$ in the domain, there are, in fact, many of these local degree $d$ functions defined on $x$. 

\subsection{The Sampling Lemma}
In this section, we show a sampling lemma regarding hyperplanes and points.  The proof is similar to that of \cite[Lemma 5.1]{MZpcp}, but we strengthen the analysis. In particular, \cref{lm: sampling points v2} below includes an additive ``error'' term of $\sim q/M$, which improves upon a similar factor with $\sqrt M$ in the denominator in \cite[Lemma 5.1]{MZpcp}.

We next describe the sampling process and the lemma. Let $\mc{W} = \{W_1,\ldots, W_M\}$ be a collection of hyperplanes in $\F_q^n$. 
Denote by $\mu$ the uniform measure over $x\in \F_q^n$, i.e., $\mu(x)=q^{-n}$, and by $\nu_{\mc{W}}$ the probability measure over $\F_q^n$ generated by the following sampling procedure:
\begin{enumerate}
    \item Choose $W_i \in \mc{W}$ uniformly at random.\label{itm: random hyperplane}
    \item Choose a point $x\in W_i$ uniformly at random.\label{itm: uniform point in hyperplane}
\end{enumerate}
We will omit the subscript $\mc{W}$ and use $\nu$ when it is clear from the context. Furthermore, we naturally extend the measure on point to sets by $\mu(S) = \sum_{x\in S} \mu(x)$ and similarly for $\nu$. Let us define for $x \in \F_q^n$ the set of hyperplanes containing it 
\begin{equation}
    \label{eq: number of containing hyperplanes}
    N(x) := |\{W_i \in \mc{W} \; | \;   W_i\ni x \}|.
\end{equation}
Since \cref{itm: random hyperplane}  samples $W_i\ni x$ with probability $\frac{N(x)}{M}$, and conditioned on this event, \cref{itm: uniform point in hyperplane} chooses $x$ with probability $\frac{1}{|W_i|}=q^{-(n-1)}$, we get that by definition 
\begin{equation}\label{eq: sampling measure formula}
    \nu(x)=\frac{N(x)}{M}\cdot q^{-(n-1)}.
\end{equation}

We turn to show the following lemma, which asserts that the measures $\mu$ and $\nu_{\mc{W}}$ are close whenever $\mc W$ contains sufficiently many hyperplanes.

\begin{lemma}[Sampling] \label{lm: sampling points v2}
For any collection $\mc{W} = \{W_1,\ldots, W_M\}$ of $M$ distinct hyperplanes in $\F_q^n$, and any set of points $S \subseteq \F_q^n$, the probability measures $\mu,\nu_{\mc{W}}$ satisfy 
  \[
 \frac{1}{2}\left(\mu(S) - \frac{4q}{M}\right) \leq \nu_{\mc{W}}(S) \leq 2\mu(S) + \frac{8q}{M}.
  \]
\end{lemma}

The key to proving the sampling lemma is analyzing the concentration of $N(x)$.

\begin{claim} \label{clm: chebyshev sampling}
    For any $c>0$ it holds that
    \[
    \Pr_{x\in \F_q^n}\left[ \left|N(x) - \frac{M}{q} \right| \geq \frac{c \cdot M}{q} \right] \leq \frac{q}{c^2  M},
    \]

\end{claim}
\begin{proof}
We analyze the random variable $N(x)$ where $x$ is uniformly chosen from $\F_q^n$, showing that  
\[
\E[N(x)] = \frac{M}{q} \quad \text{and} \quad \var[N(x)] \leq \frac{M}{q}.
\]
The lemma then follows by applying Chebyshev's inequality. 

For the expectation, we write $N(x)$ as a sum of indicators, and use linearity of expectation 
    \[
        \E_{x}[N(x)] 
        = \sum_{i=1}^M \E_{x}[\ind{x\in W_i}]
        = \sum_{i = 1}^M \Pr_{x}[x\in W_i] 
        = \frac{M}{q}, 
    \]
    where the last equality follows from the cardinality of a hyperplane in $\F_q^n$.
    
    For the variance, writing $N(x)$ as a sum, we get
    \begin{align*}
        \var[N(x)]
        &=\sum_{i,j\in[M]}cov(\ind{x\in W_i},\ind{x\in W_j})\\
        &=\sum_{i\in[M]}\var[\ind{x\in W_i}]+\sum_{i\neq j\in[M]}cov(\ind{x\in W_i},\ind{x\in W_j}) .
    \end{align*}
    The first summation is at most $M/q$, as for each $i\in[M]$ we have $\var[\ind{x\in W_i}] = 1/q - 1/q^2\le 1/q$.
    We next show the second summation equals $0$, using the premise that all hyperplanes are distinct. Indeed, for any $i\neq j$ the intersection of $W_i \neq W_j$ is a subspace of dimension $n-2$ and cardinality $M/q^2$. Thus, for any $i\neq j$ we get
    \begin{align*}
        cov(\ind{x\in W_i},\ind{x\in W_j})
        &=\E[\ind{x\in W_i}\cdot \ind{x\in W_j}] - \E[\ind{x\in W}] \cdot \E[\ind{x\in W'}]\\
        &=\Pr_x[x\in W_i \cap W_j] - \Pr_x[x\in W_i] \cdot \Pr_x[x\in W_j]\\
        &= 1/q^2 - (1/q)^2 = 0 .
    \end{align*}    
    Overall, we get that $\var[N(x)] \leq M/q$, as required.
\end{proof}

We are now ready to prove the sampling lemma.
\begin{proof}[Proof of \cref{lm: sampling points v2}]
 For each integer $i$, let 
\[
    m_i :=  \left|\left\{x \in \F_q^n \; | \; \frac{2^i\cdot M}{q} \leq N(x) < \frac{2^{i+1}\cdot M}{q}\right\} \right| .
\]
By \cref{clm: chebyshev sampling}, we get that
\[
\frac{m_i}{q^n} \leq \frac{q}{(2^i-1)^2 \cdot M}.
\]
To complete the picture, let $m_0 = \card{\set{x\in \F_q^n \; | \; N(x) \leq M/q}}$. 
Towards the upper bound, we expand $\nu_{\mc{W}}(S)$ and perform a dyadic partitioning with $m_0 \leq \card{S}$:
\begin{align*}
    \nu_{\mc{W}}(S) = \sum_{x \in S} \nu_{\mc{W}}(x) 
    = \sum_{x\in S} \frac{N(x)}{M\cdot q^{n-1}} 
    \leq  \frac{1}{M\cdot q^{n-1}} \left(\frac{2\cdot M \cdot |S|}{q} + \sum_{i = 1}^{\infty} \frac{m_i \cdot 2^{i+1} \cdot M}{q} \right).
\end{align*} 
Once we multiply, the first term in the parenthesis contributes $2\mu(S)$, whereas the contribution of
the second summand can be upper bounded by
\[
    \frac{1}{q^n}\cdot \sum_{i = 1}^{\infty} m_i \cdot 2^{i+1}
    \leq \frac{q}{M}\cdot \sum_{i=1}^{\infty} \frac{2^{i+1}}{(2^i-1)^2} 
    \leq \frac{8q}{M} , 
\]
where the last inequality is since for all $i\geq 1$, it holds that $\frac{2^{i+1}}{(2^i-1)^2}\leq \frac{8}{2^i}$, and the infinite sum of the latter is at most $8$.

For the lower bound, recall $\nu_{\mc{W}}(x)=\frac{N(x)}{M\cdot q^{n-1}}$ from \eqref{eq: sampling measure formula}. We expand
\[
    \nu_{\mc{W}}(S) 
    = \sum_{x\in S}\nu_{\mc{W}}(x)
    \geq  \sum_{\set{x\in S : N(x) \geq M/(2q)}}\nu_{\mc{W}}(x) \geq |\{x\in S \; | \; N(x) \geq  M/(2q) \}| \cdot \frac{1}{2q^{n}} .
    \]
We next use \cref{clm: chebyshev sampling} with the parameters above, which gives
\[
\Pr_{x}\left[N(x) < \frac{M}{2q} \right] \leq \frac{4q}{ M}.
\]
It follows that 
\[
|\{x\in S\; | \; N_{\mc{W}}(x) \geq M/(2q) \}| \geq q^n\left((\mu(S) - \frac{4q}{M}\right).
\]
Putting everything together, we get the desired lower bound:
\[
\nu_{\mc{W}}(S) \geq \frac{1}{2}\left(\mu(S) - \frac{4q}{M}\right).
\]
\end{proof}

\subsection{Proof of \texorpdfstring{\cref{lm: hyperplane general}}{the Hyperplane Agreement Lemma}}

Suppose $\mc{W} = \{W_1, \ldots, W_M\}$ is the set of distinct hyperplanes on which $f$ is $\varepsilon$-close $\mc{C}_{n-1}$. For each $W_i$, let $f_i \in \mc{C}_{n-1}$ be this function, so 
\[
\dist(f_i, f|_{W_i}) \leq \varepsilon,
\]
for each $1 \leq i \leq M$. Note that if $\eps \geq 1/4$  then the result is trivial, so for the remainder of the proof we suppose $\eps < \frac{1}{4}$.

We now construct the following graph $G = (V, E)$, typically referred to as the \emph{consistency graph}. The vertex set is $V = \{1, \ldots, M\}$ while the set of edges $E$ consists of all pairs $(i,j)$ such that the functions $f_i$ and $f_j$ are consistent: 
\[
f_i|_{W_i \cap W_j} = f_j|_{W_i \cap W_j}.
\]
If $W_i \cap W_j$ is empty, then we also consider $f_i$ and $f_j$ consistent and have $(i,j) \in E$. We will first show that $G$ contains a large clique. Call a graph \emph{transitive} if it is an edge disjoint union of cliques. Define the following ``measure'' of non-transitivity:
\[
\beta(G) = \max_{(i,j) \notin E} \Pr_{k \in V}[(i,k), (j,k) \in E].
\]
Clearly, a graph $G$ is transitive if and only if $\beta(G) = 0$. The following lemma from \cite{RazS97} gives an approximate version of this, quantifying how close is $G$ to be transitive in terms of $\beta(G)$.
\begin{lemma}\cite[Lemma 2]{RazS97} \label{lm: transitive subgraph}
A graph $G$ can be made transitive by removing at most $3\sqrt{\beta(G)}|V|^2$ edges.
\end{lemma}

In our next two lemmas we show that $G$ has many edges and then bound $\beta(G)$. Along with \cref{lm: transitive subgraph}, these steps will establish that $G$ contains a large clique.
The large clique corresponds to many local degree $d$ functions $f_i$ that are consistent with one another, and using these functions we will extrapolate a degree $d$ function that is close to $f$.

\begin{lemma} \label{lm: many edges}
The graph $G$ has at least $0.4M^2$ edges.
\end{lemma}
\begin{proof}
    For each $W_i$, let $S_i \subseteq W_i$ be the set of points on which $f_i$ and $f$ differ. We have that
    \[
    \frac{|S_i|}{|W_i|} \leq \varepsilon,
    \]
    by assumption.
    
    Now fix a $W_i$. We will apply \cref{lm: sampling points v2} inside of $W_i$. For each $j$, let $W'_j = W_i \cap W_j$. Let $\mc{W}' = \left\{W'_j \; | \; \frac{W'_j \cap S_i}{W'_j} \geq 3 \varepsilon \right\}$, let $\mc{I} = \{j \in [M] \;| \; W'_j \in \mc{W}'\}$, and say $M' = |\mc{W'}|$. Now consider the size of $S_i$ inside of $W_i$ under the measure $\nu_{\mc{W'}}$. We have,
    \[
    \nu_{\mc{W'}}(S_i) \geq 3 \varepsilon.
    \]
     By \cref{lm: sampling points v2}
    \[
     \nu_{\mc{W}'}(S_i) \leq 2\eps + 5q/M'.
    \]
   Combining both bounds, it follows that,
    \[
    M' \leq \frac{8q}{\eps}.
    \]
    Finally, we can bound $|\mc{I}| \leq 8q^3/\eps$ by noting that for any $W' \in \mc{W}'$, there can be at most $q^2$ hyperplanes $W$ such that $W \cap W_i = W'$. 
    
     Thus, for a fixed $i$,
    \[
    \Pr_{W_j}\left[\frac{|W_i \cap W_j \cap S_i|}{|W_i \cap W_j|} \geq 3 \varepsilon\right] \leq \Pr_{j \in [M]}[j \in \mc{I}] \leq \frac{8q^3}{\eps M}.
    \]
    Now choose $W_i, W_j$ uniformly at random. By a union bound, it follows that, with probability at least $1 - \frac{16q^3}{\eps M}$, we have that both
    \[
    \frac{|W_j \cap W_i \cap S_i|}{|W_j \cap W_i|} \leq 3 \varepsilon \quad \text{and} \quad   \frac{|W_j \cap W_i \cap S_j|}{|W_j \cap W_i|} \leq 3 \varepsilon.
    \]
    For such $i, j$, we get that 
    \[
    \Pr_{x \in W_i \cap W_j}[f_i(x) = f_j(x)] \geq 1 - \Pr_{x \in W_i \cap W_j}[x \in S_i \cup S_j] \geq 1 - 6 \varepsilon > 1 - \delta_0,
    \]
     so $f_{i}|_{W_i \cap W_j} = f_j|_{W_i \cap W_j}$ by \eqref{eq: min distance over lifted codes}. Thus, choosing $i,j$ randomly, we get that they are adjacent in $G$ with probability at least $1 - \frac{16q^3}{\eps M} \geq 0.9$.
\end{proof}

\begin{lemma} \label{lm: bound beta}
    Suppose $W_i, W_j$ are distinct  hyperplanes such that $f_{i}|_{W_i \cap W_j} \neq f_{j}|_{W_i \cap W_j}$. Then,

     \[
    \Pr_{W_{k} \in \mc{W}}\left[f_{i}|_{W_i \cap W_k} \equiv f_{k}|_{W_i \cap W_k}\bigwedge f_{j}|_{ W_j \cap W_k}\equiv f_{k}|_{ W_j \cap W_k}\right] \leq \frac{1}{100}.
    \]
\end{lemma}
\begin{proof}
    Let $W = W_i \cap W_j$. For each $k$ let $W'_k = W_k \cap W$, let $f'_i = f_i|_W$, and let $f'_j = f_j|_W$. Also let 
    \[
    \mc{W}' = \{W \cap W_k \; | \; W_k \in \mc{W} \setminus \{W_i, W_j\}, f'_i|_{W'_k} = f'_j|_{W'_k} \},
    \]
    and let $\mc{I} = \{k \in [M] \; | \; W \cap W_k \in \mc{W}'\}$. Note that the probability of interest is at most,
    \[
    \Pr_{W_{k} \in \mc{W}}[f_{i}|_{W_i \cap W_j \cap W_k} = f_{j}|_{W_i \cap W_j \cap W_k}] = \frac{|\mc{W}'|}{|\mc{W}|} = \frac{|\mc{I}|}{M} \leq \frac{q^3 |\mc{W}'|}{M}, 
    \]
and therefore we may focus on bounding $M' := |\mc{W}'|$. Above we are using the fact that for each $W' \in \mc{W'}$, there can be at most $q^3$ values $k$ such that $W' = W \cap W_k$. This follows by noting that the dual space of $W'$ has size $q^3$ and the one-dimensional dual space of $W_k$ must be inside this space.
    
To bound $M'$, we will consider $W$ as our ambient space and apply \cref{lm: sampling points v2} inside of $W$, where the set of sampling hyperplanes is $\mc{W}'$.  Let $S \subseteq W$ be
    \[
    S = \{x \in W \; | \; f'_i(x) \neq f'_j(x) \}.
    \]
The key observation is that sampling $W' \in \mc{W}'$ and then a point $x \in W'$, the probability that $x \in S$ is actually $0$. On the other hand, $S$ is rather large since $f'_i, f'_j$ are both from a code with non-trivial distance, so together with \cref{lm: sampling points v2}, these two observations give an upper bound on the size of $S$.

Indeed, by the definition of $\mc{W}'$, we have that $\nu_{\mc{W}'}(S) = 0$, while on the other hand, $|S|/|W| \geq \delta_0$ by the assumption that $f'_i \neq f'_j$ combined with the definition of $\delta_0$ from \eqref{eq: min distance over lifted codes}. Now  by \cref{lm: sampling points v2}, we have
    \[
    0 \geq \frac{1}{2}\left(\delta_0- \frac{4q}{M'} \right).
    \]
    Thus,
    \[
    M' \leq \frac{4q}{\delta_0}.
    \]
Since $\eps\leq \delta_0/6$ and $M\geq 10^5 q^4/\eps\geq 10^5 q^4/\delta_0$, we conclude that $q^3 M'/M\leq 1/100$ and the lemma follows.
\end{proof}

\cref{lm: bound beta} shows that $\beta(G) \leq 1/100$, and applying \cref{lm: transitive subgraph}, $G$ can be made transitive by removing at most $0.3M^2$ edges. Combined with \cref{lm: many edges}, the remaining transitive subgraph of $G$ has at least $0.1M^2$ edges.
Let $\mc{K}_1, \ldots, \mc{K}_J$ be the edge disjoint union cliques of this transitive subgraph and say $\mc{K}_1$ is the largest one. It follows that,
\[
0.1M^2 \leq \sum_{I=1}^J |\mc{K}_i|^2/2 \leq |\mc{K}_1| \cdot M/2,
\]
and $G$ contains a clique of size $M_1 := \card{\mc{K}_1} = 0.2M$. From now on we let $\mc{K}$ denote this clique and without loss of generality let us suppose the members are $\mc{K}= \{1, \ldots, M_1\}$. We define the following function $F$ using the functions $f_1, \ldots, f_{M_1}$: set $F(x) = f_j(x)$ if there is some $j \in \mc{K}$ such that $x \in W_j$, and $F(x) = 0$ otherwise.

We are now ready to prove the main lemma of this section,
\begin{proof}[Proof of \cref{lm: hyperplane general}]
 We will show that 
\begin{enumerate}
    \item \label{itm: global F close to f} $F$ is close to $f$.
    \item\label{itm:global F close to low degree} $F$ is close to a degree $d$ function.
\end{enumerate} 
Using the triangle inequality, this will establish that $f$ is close to a degree $d$ function.

\paragraph{\cref{itm: global F close to f}: Proof that $F$ is $3\eps$-close to $f$.}
Let $S = \{x \in \F^n \; | \; f(x) \neq F(x) \}$. Since $f$ is $\eps$-close to $F$ on all hyperplanes in $\mc{H}:=\{W_i\}_{i \in \mc{K}}$, we have
\[
\nu_{\mc{H}}(S) \leq \eps.
\]
By \cref{lm: sampling points v2}, it follows that
\[
\eps \geq \frac{1}{2}\left(\mu(S) -  \frac{4q}{0.6M}\right),
\]
and by our assumption on $M$
\[
\mu(S) \leq 3\eps.
\]

\paragraph{\cref{itm:global F close to low degree}: Proof that $F$ is $\eps$-close to $\mc{C}_n$.}
We will show that $F$ passes the $t$-space test with high probability. Recall that this test operates by choosing a random subspace $L$ of dimension $t$ and accepting if and only if $F|_L \in \mc{C}_t$. By \cref{def: t test assump} we have,
\begin{equation}\label{eq: rejection lb}
    \Pr_{L \in \mc{G}_t}[F|_L \notin \mc{C}_t] \geq c(\dist(F,\mc{C}_n)).
\end{equation}
Next, we upper bound the rejection probability to bound the distance of $F$ from the family of functions $\mc{C}_n$. 
Let $\mc{L} = \{L \in \mc{G}_t \; | \; \nexists i \in \mc{K}, W_i \supseteq L \}$. 

We have
\[
\Pr_{L \in \mc{G}_t}[F|_L \notin \mc{C}_t] \leq \Pr_{L \in \mc{G}_t}[\nexists i \in \mc{K}, W_i \supseteq L].
\]
We will follow similar calculations to those of \cref{clm: chebyshev sampling}. 
Here, for any $L\in\mc{G}_t$, we let 
\[N(L)=|\{i \in \mc{K} \; | \;  W_i\supseteq L \}|.\]
We can now express the upper bound using $N(\cdot)$. 
\[\Pr_{L \in \mc{G}_t}[\nexists i \in \mc{K}, W_i \supseteq L]=\Pr_{L \in \mc{G}_t}[N(L)=0].\]
To further bound the probability, we next bound the first and second moments of the random variable $N(L)$ for uniform random $L\in \mc{G}_t$. 
Let 
\[p_1 := \Pr_{L \in \mc{G}_t}[L \subseteq W].
\]
 For the expectation, note that by the linearity of expectation 
    \[
    \E_{L \in \mc{G}_t}[N(L)] = \sum_{i = 1}^M \Pr_{L \in \mc{G}_t}[L \subseteq W_i] = p_1M.
    \]
    For the variance analysis, similarly to $p_1$ we define the probability of a uniform subspace to be contained in a fixed $n-2$ dimensional subspace, which we will use as an intersection of two distinct hyperplanes $W,W'$.
    \[p_2 := \Pr_{L \in \mc{G}_t}[L \subseteq W \cap W'].
    \] 
    We next bound the $\E_{L \in \mc{G}_t}\left[N(L)^2\right]$:
\begin{align*}
   \E_{L \in \mc{G}_t}\left[N(L)^2 \right] &= \E_{L \in \mc{G}_t} \left[\left(\sum_{i=1}^M \ind{L \subseteq W_i} \right)^2 \right] 
   \\
   & \leq M \cdot \Pr_{L \in \mc{G}_t}[L \subseteq W_i] + M^2 \cdot \Pr_{L \in \mc{G}_t}[L \subseteq W_i \cap W_{i'}]\\
   &= p_1M + p_2M^2\\
   &\leq p_1M+p_1^2M^2.
\end{align*} 
The last inequality follows from 
\[p_2=\frac{\qbin{n-2}{t}}{\qbin{n}{t}}\leq \left(\frac{\qbin{n-1}{t}}{\qbin{n}{t}}\right)^2=p_1^2.\]

Now,
\begin{align*}
    \Pr_{L \in \mc{G}_t}[N(L)=0]\le \frac{\var(N(L))}{\E[N(L)^2]}=1-\frac{\E[N(L)]^2}{\E[N(L)^2]}\leq 1-\frac{p_1^2M^2}{p_1M+p_1^2M^2}\leq \frac{1}{1+p_1M}\le \frac{1}{p_1M}.
\end{align*}
Note that, 
\[p_1=\frac{\qbin{n-1}{t}}{\qbin{n}{t}}=\frac{\prod_{i=0}^{t-1}(q^{n-1}-q^i)}{\prod_{i=0}^{t-1}(q^{n}-q^i)}=q^{-t}\frac{\prod_{i=1}^{t}(q^{n}-q^i)}{\prod_{i=0}^{t-1}(q^{n}-q^i)}=q^{-t}\frac{q^n-q^t}{q^n-1}\geq \frac {1}{2q^{t}}.\]
By our assumption on the size of $M$,
\[
 \frac{1}{p_1 \cdot M}\le \frac{2q^{t}}{M} \le c(\eps).
\]
Combining with \eqref{eq: rejection lb}, and recall that $c(\cdot)$ is non-decreasing, we conclude $\dist(f, \mc{C}_n) \leq \eps$.

\end{proof}

\section{Sample-Based Testers for All Affine Invariant Families}

In this section, we describe our sample-based tester that works for any fixed 
family of functions $\mc{C}_n \subseteq \{\F^n \to \F \}$. The tester samples a certain number of points and makes a decision based on this view. Since we wish the tester to have a one-sided error (always accept inputs from the family $\mc{C}_n$), the algorithm does the only thing possible. It rejects only when there is no codeword that can explain the sampled view. 

\paragraph{The $\mathcal{T}$ for $\mc{C}_n$:}
\begin{enumerate}
    \item Sample $s$ points from $\F^n$ uniformly with replacement. Let $S$ denote the set of points sampled.
    \item Accept if there is a function $g \in \mc{C}_n$ satisfying $g|_S = f|_S$. Otherwise, Reject.
\end{enumerate}

\begin{lemma}\label{lm: sample-based soundness with epsilon parameter}
    Let $f$ be $\eps$-far from $\mc{C}_n$, if $s\ge \frac{\ln (2|\mc{C}_n|)}{\eps}$ then $\mathcal{T}$ rejects with probability at least  $1/4$.
\end{lemma}
\begin{proof}
    For any $g\in\mc{C}_n$, we have $\dist(f,g)\ge \eps$ and therefore the probability $g$ is consistent with $f$ on $s_n$ uniform random samples is at most  $(1-\eps)^{s_n}$. Union bounding over all members of $\mc{C}_n$, the probability that there exists  $g\in \mc{C}_n$ such that $f|_S\equiv g|_S$, which is the case of acceptance, is at most
\[
(1-\eps)^{s}|\mc{C}_n| \leq \exp(-\eps s)|\mc{C}_n|\le 1/2.
\]
The last inequality is by the bound on $s$.
\end{proof}
Though the result of \cref{lm: sample-based soundness with epsilon parameter} is easy and quite general, it suits well when the tester receives the proximity parameter $\eps$ as an input and can decide how many samples to make based on the proximity parameter. 

Since we are going to apply the sample-based tester on functions with varying distances, and yet we need to squeeze all the probability of rejecting non-codewords. We will use an oblivious (to the proximity parameter) algorithm; the only difference is the setting of the number of samples it uses. We show that such a tester is optimal in the sense discussed in \cref{par:optimal testing}. 

Our approach to the analysis now is based on the \emph{learn-and-test} paradigm. The tester operates by choosing twice the amount of points needed to interpolate a unique function from $\mc{C}_n$; crucially, this amount is independent of the proximity parameter $\eps$. The first half of the points are used to reduce to one candidate function $g \in \mc{C}_n$, and the second half of the points are used to check if the tested function is indeed $g$. 

\begin{lemma}\label{lm: sample-based lifted codes soundness}
    Let $f$ be $\eps$-far from $\mc{C}_n$, if $s\ge \frac{10\ln |\mc{C}_n|}{\delta}$, where $\delta=\delta(\mc{C}_n)$ is the minimal distance of the code $\mc{C}_n$, then $\mathcal{T}$ rejects with probability at least $\min\{1/4,s \eps/8\}$.
\end{lemma}
\begin{proof}

Let $S_1$ be the first $\frac{5\ln|\mc{C}_n|}{\delta}$ points sampled in $S$, and let $S_2$ be the remaining points. Let $E$ be the event that there is at most one $g \in \mc{C}_n$ satisfying $g|_{S_1} = f|_{S_1}$. Then the probability that the tester rejects is at least 
\[
\Pr[E] \cdot \Pr_{S}[g|_{S_2} \neq f|_{S_2} \; | \; E],
\]
where in the probability $g$ refers to the unique function in $\mc{C}_n$ agreeing with $f$ on $S_1$. By \cref{claim:sample_based_ub general}, 
\[
\Pr_{S}[\overline{E}] \leq  \exp(-\delta |S_1|)|\mc{C}_n|^2 \leq \frac{1}{2}.
\]
Furthermore, by the assumption that $f$ is $\eps$-far from $\mc{C}_n$, we have,
\[
\Pr_{S}[g|_{S_2} \neq f|_{S_2} \; | \; E] \geq 1 - (1-\eps)^{s/2} \geq \min\bigg\{\frac{1}{2},\frac{s \eps}{4}\bigg\}.
\]
The last inequality is proved in the proof of \cite[Lemma 3.1]{ben2024property}. The result follows. 
\end{proof}

We will use \cref{lm: sample-based lifted codes soundness} in order to build our semi-sample-based tester. Here we note a quick application of it, improving on the sample complexity given by \cref{lm: sample-based soundness with epsilon parameter} whenever $\eps$ is much smaller than $\delta_0$.  
\begin{theorem}\label{thm: oblivious sample-based tester}
    There is a sample-based tester for $\mc{C}_n$ with a one-sided error that uses at most $O(s+1/\eps)$ samples. Where $s=O(\ln|\mc{C}_n|/\delta)$ and $\delta=\delta(\mc{C}_n)$ is the minimal distance between two distinct codewords in $\mc{C}_n$.
\end{theorem}
\begin{proof}
    We claim that repeating $\mc{T}$ (with $s=O(\ln|\mc{C}_n|/\delta)$) enough times and rejecting if any of the instances rejects is a valid tester. Indeed, if $f$ is a codeword, all samples will be consistent with $f$, and we always accept. If $f$ is $\eps$-far, then by \cref{lm: sample-based lifted codes soundness} each instance accepts with probability at most $1-\min\{1/4,s \eps/8\}$. Since the instances are independent, by repeating $O\left(\frac{1}{s\eps}+1\right)$ times, the probability that they all accept is at most $1/3$.
    For the sample complexity, we have any instance use $s$ samples, which in total become $O\left(s\left(\frac{1}{s\eps}+1\right)\right)=O(s+1/\eps)$.
\end{proof}

 \section{Semi-Sample-Based Testers}\label{sec: semi-sample-based tester}

The semi-sample-based tester for $\mc{C}_n$ on input function $f: \F_q^n \to \F_{q}$ works by first choosing a random $k$-dimensional subspace, $A$, and then running the sample-based tester on the smaller, restricted function $f|_A$.

\paragraph{The tester $\mathcal{T}_k(Q, \mc{C}_n)$:}
\begin{enumerate}
    \item Choose a $k$-dimensional subspace $A\subset \F^n$ uniformly at random.
    \item Sample $Q$ points from $A$ uniformly with replacement. Let $S$ denote the set of points sampled.
    \item Accept if there is a function $g \in \mc{C}_k$ satisfying $g|_S = f|_S$. Otherwise, Reject.
\end{enumerate}   

We will show that with large enough $k$, choosing $Q$ sufficiently large with respect to the rate and distance of $\mc{C}_k$, the tester $\mc{T}_k(Q)$ is indeed a valid tester with perfect completeness and non-trivial soundness (or even optimal soundness for certain families). 

\begin{remark}
We remark that while our focus is on query complexity rather than time complexity, the tester $\mc{T}_k(Q, \mc{C}_n)$ does have good running time when $\mc{C}_n$ is a linear code. In this case, checking if there exists $g \in \mc{C}_k$ such that $g|_S = f|_S$ can be done by solving a system of $|S|$ linear equations with $\log_q|\mc{C}_k| \leq |S|$ variables. In particular, the time complexity is $O(|S|^3)$.
\end{remark}

In this section, we show our main theorem regarding semi-sample-based testers for affine invariant families of functions. At a high level, our tester for \emph{every} affine invariant family is simply the tester $\mc{T}_k(Q)$ except one must choose the quantity $k$ as well as the number of queries $Q$ depending on the specific code. For families of functions such as Reed-Muller codes or lifted affine-invariant codes, we are able to give explicit values of $k$ and $Q$, and show that the semi-sample-based tester achieves optimal soundness. For general affine invariant families however, the known soundness of current subspace based testers is not strong enough for us to give explicit values of $k$ and $Q$ for a valid semi-sample-based tester. We state our main theorems for testing Reed-Muller codes, lifted affine-invariant codes, and finally general affine invariant families below.

We first state our main result for Reed-Muller codes. The proof is deferred to \cref{sec: rm test proof}.
\begin{theorem} \label{th: semi-sample rm}
    Let $f: \F^n \to \F$, let
    $k \geq 8d + 3\KMconst + 150$
    be a dimension parameter,  let $\delta_0 := q^{-d/(q-1)}$ and set
    \[
  s_k := \ceil{\frac{100\ln |\RM[k,q,d]|}{\delta_{0}}}+ 1 
    \]
    Then $\mc{T}^f_k(s_k, \RM[n,q,d])$ is a valid tester for the degree $d$ Reed-Muller code over $\F^n$. Specifically, the tester satisfies the following:
    \begin{itemize}
        \item Completeness: If $f \in \RM[n,q,d]$, then the tester accepts with probability $1$.
        \item Soundness: If $f$ is $\eps$-far from $\RM[n,q,d]$ then the tester rejects with probability at least $\min\{s_k \eps/8, 1/128\}$
    \end{itemize}
   Furthermore, the tester has query complexity $s_k$, which is bounded by
   \begin{equation}\label{eq:asymptotic_bound_for query_complexity_SSB}
       s_k = O\left(q^{t_{q,d}} \cdot \left(\frac{e(d+k)}{k}\right)^k \right)\ .
   \end{equation}
\end{theorem}

For lifted affine invariant codes, $\mc{C}_n = \mc{C}^{t \nearrow n}$, we give semi-sample-based testers provided there is a dimension $k$ where the size of the restricted code, $|\mc{C}_k|$, is at most $O\left(q^k \delta_0/\log(q) \right)$. The proof is laid out in \Cref{sec: lifted test proof}. 
\begin{theorem} \label{th: semi-sample lifted}
    Let $\mc{C}_n = \mc{C}^{t\nearrow n}$ be the $n$-dimensional lift of the affine invariant code $\mc{C}_t \subseteq \{\F^t \to \F\}$ and let $\delta_0$ be the smallest minimum distance of the codes $\mc{C}_t, \cdots, \mc{C}_n$. Suppose $k$ is a dimension parameter satisfying
    \[
    \frac{\log_q |\mc{C}_k|}{q^k} \leq \frac{\delta_0}{10^6 \log(q)},
    \]
    and set
    \[
    Q_k := \ceil{\frac{100\ln |\mc{C}_k|}{\delta_{0}}}+ 1.
    \]
    Given an input function $f: \F^n \to \F$, the tester $\mc{T}^f_{k}(Q_k, \mc{C}_n)$ is a valid tester for the lifted affine-invariant code, $\mc{C}_n$. Specifically, the tester satisfies the following:
    \begin{itemize}
        \item Completeness: If $f \in \mc{C}_n$, then the tester accepts with probability $1$.
        \item Soundness: If $f$ is $\eps$-far from $\mc{C}_n$, then the tester rejects with probability at least $$\min\{Q_k \eps/8, 1/128 \}.$$
    \end{itemize}
\end{theorem}

Finally, we also show that there is a semi-sample tester for general affine invariant families, provided that there is an $r$-flat tester with some known $r$ that has strong enough soundness. The proof is outlined in \Cref{sec: general test proof}. 

\begin{theorem}\label{thm:semi-sample-based soundness general}
Let $\mc{C}_n$ be an affine invariant function that is locally testable by the $r$-space test with soundness function $c(\cdot)$ and let $f: \F^n \to \F$ be an input function. Let $\delta_0$ be the smallest minimum distance of the codes $\mc{C}_r, \ldots , \mc{C}_n$, and let 
\[
s_k: = \ceil{\frac{100\ln |\mc{C}_k|}{\delta_0}} + 1.
\]
Suppose that for sufficiently large enough dimension $k$ we have
\begin{equation} \label{eq: c assumption gen aff}
c\left( \frac{1}{8s_k} \right) \geq \Omega(q^{r-k}) \quad \text{and} \quad \frac{s_k}{q^k} \leq q^{-\Theta(1)}.
\end{equation}
Then, fixing such a $k$, the tester $\mc{T}^f_{k}$ for $\mc{C}_n$ has the following properties
    \begin{itemize}
        \item Completeness: If $f \in \mc{C}_n$, then the tester accepts with probability $1$.
        \item Soundness: If $f$ is $\eps$-far from $\mc{C}_n$ then 
        \[
        \Pr\big[\mc{T}^f_{k} \big] \geq \min\bigg\{\frac 1{128}, \frac{s_k \eps}{8} \bigg\}
        \]
    \end{itemize}
   Furthermore, the query complexity of the tester is $s_{k}$.
\end{theorem}

Currently we are not aware of any affine invariant families other than lifted codes and Reed-Muller codes that satisfy \eqref{eq: c assumption gen aff}. In particular, the soundness results for testing affine invariant families in \cite{bfhhl} is inexplicit. However, we are able to give semi-sample-based testers that have non-trivial soundness without assuming \eqref{eq: c assumption gen aff}. The soundness there is non-trivial but still non-explicit. We defer this result to \cref{appendix: gen aff}. 

\subsection{Reed-Muller Codes: Proof of \texorpdfstring{\cref{th: semi-sample rm}}{Semi-sample-based Tester Theorem}} \label{sec: rm test proof}
As is usual in property testing, the completeness is clear.
To bound the query complexity, $s_k$, we recall that $t_{q,d} = \lceil \frac{d}{q-q/p} \rceil + 1$, and the bounds from \cref{fact: rm distance}, which in particular imply $\ln |\RM[k,q,d]|\le \binom{d+k}{k}\cdot \ln q $, and $1/\delta_0 \le q^{t_{q,d} - 1} \leq q^{t_{q,d}} / \ln q$. Overall, this bounds the query complexity by
\begin{equation} \label{eq: s bound}
s_k = \ceil{\frac{100\ln |\RM[k,q,d]|}{\delta_{0}}}+ 1 \leq 100q^{t_{q,d}} \cdot \binom{d+k}{k}. 
\end{equation}
\cref{eq:asymptotic_bound_for query_complexity_SSB} follows from standard bounds on binomial coefficients.

From here on we focus on the proof of soundness. Fix any dimension parameter $k$ such that
\begin{equation}
    \label{eq:RM:choice of k}
    k \geq 8d + 3\KMconst + 150
\end{equation}
as in \cref{th: semi-sample rm}.
To simplify notation, we also drop the subscript from $s_k$, writing $s = s_k$ instead, and write $\mc{T}^f_k$ in place of $\mc{T}^f_k(s_k, \RM[n,q,d])$. Also let $\eps_f := \dist(f, \RM[n,q,d])$ be the actual distance from $f$ to $\RM[n,q,d]$ and we assume that $f$ is $\eps$-far from $\RM[n,q,d]$, meaning $\eps_f \geq \eps$. Finally, we remark that by \cref{fact: rm distance}, $\delta_0$ is the minimum distance of $\RM[t,q,d]$ for all $t_{q,d} \leq t \leq n$.

Similar to the soundness analysis of the standard subspace test for Reed-Muller codes in \cite{BKSSZ10}, we split our soundness analysis into two cases based on if the magnitude of $\eps_f$ is small or large. The analyses of these two cases together will show the desired soundness for \cref{th: semi-sample rm}.

Before going into the proof, we state and prove two key claims. The first claim shows that $s_k$ is sufficiently smaller than $q^k$.

\begin{claim}\label{fact: sk small}
    Fix any constant $c > 0$. For any $k \geq 8d + 3c + 24$, it holds that
    \[
        \frac{s_k}{q^k} \leq q^{-c}.
    \]
    
\end{claim}
\begin{proof}
    Define $B_k := 100q^{t_{q,d}} \cdot \binom{d+k}{d}$, the upper bound on $s_k$ from \eqref{eq: s bound}.
    We focus on the expression $B_k/q^k$ instead, bound it for a proper dimension $k_0$, and show it decreases multiplicatively by a significant factor as $k$ increases until for certain $k(c)$ the whole quotient is small enough, immediately implying the lemma.
    One ``step" changes the quotient by a factor of
    \begin{equation}
        \label{eq: B_k iterative bound}      \frac{\frac{B_{k+1}}{q^{k+1}}}{\frac{B_k}{q^k}}
        = \frac{1}{q}\cdot \frac{B_{k+1}}{B_k}
        = \frac{1}{q} \cdot \left(1 + \frac{d}{k+1}\right) \ ,  
    \end{equation}
    where the second equality is due to quotient of ``consecutive" binomial coefficients.

    From here on, We split into two cases. 
    \paragraph{The case $q=2$.}
    Here, it holds that $t_{q,d} = d+1$. We start from dimension $k_0 = 2d$ and plug in $q$ and $t_{q,d}$ to get:
    \[
        \frac{B_{2d}}{q^{2d}}
        \leq 100 q^{t_{q,d} - 2d} \cdot\binom{3d}{d}
        \leq 100 q^{1 - d}  \cdot 2^{3d}
        \le 2^{2d+8}\  .
    \]
    
    By \eqref{eq: B_k iterative bound}, for any choice of $k \geq 2d$ the quotient diminishes by a factor of at most 
    \[
        \frac{\frac{B_{k+1}}{q^{k+1}}}{\frac{B_k}{q^k}}
        = \frac{1}{q} \cdot \left(1 + \frac{d}{k+1}\right) < 2^{-1/3} \ .
    \]
    Applying the above for $z \geq 3(2d + c+ 8)$ steps, we get
    \[
        \frac{B_{2d+z}}{q^{2d+z}}
        \leq \frac{B_{2d}}{q^{2d}} \cdot  \left(2^{-1/3}\right)^z
        \le 2^{2d+8-z/3} 
        \leq 2^{-c} \ .
    \]
    That is, for any dimension $k = 2d+z \geq 8d + 3c + 24$ the assertion holds for $q=2$.
    
     \paragraph{The case $q > 2$.}
     Here we have $t_{q,d} \leq d$. Plugging this and $q\geq 3$, we directly calculate for $k_0 = d$:
    
    We next bound the quotient for $k=d$ (for $q > 2$ we have $t_{q,d} \leq d$):
    \[
        \frac{B_d}{q^d}
        \leq 100 q^{t_{q,d} - d} \cdot\binom{2d}{d}
        \leq 100 \cdot 4^d 
        \leq q^{2d+5} \ .
    \]
    By \eqref{eq: B_k iterative bound}, for any $k \geq d$, the quotient diminishes by a factor of
    \[
        \frac{\frac{B_{k+1}}{q^{k+1}}}{\frac{B_k}{q^k}}
        = \frac{1}{q} \cdot \left(1 + \frac{d}{k+1}\right) < \frac{2}{q} <q^{-1/3} \ .
    \]
    Applying the above iteratively for $z \geq 3(2d+c+5)$ steps, we get
    \[
        \frac{B_{d+z}}{q^{d+z}}
        \leq \frac{B_{d}}{q^{d}} \cdot  q^{-z/3}
        \leq q^{2d + 5 - z/3}
        \le q^{-c} ,
    \]
    concluding that for any dimension $k = d+z \geq 7d + 3c + 15$ the assertion holds for $q > 2$.
   
\end{proof}

The second claim is a specialization of \cref{lm: hyperplane general} to degree $d$ Reed-Muller codes with the parameters that we will use. 

\begin{claim} \label{lm: hyperplane deg d}
Suppose there are 
$M = 10^6 q^{\max\set{4, \KMconst}}\cdot s$ distinct
hyperplanes $W_1, \ldots, W_M \subset \F^n$ such that for each $i = 1, \ldots, M$, $f|_{W_i}$ is $\frac{1}{8s}$-close to a degree $d$ function on $W_i$. Then $f$ is $\frac{1}{2s}$-close to a degree $d$ function.
\end{claim}
\begin{proof}
We apply \cref{lm: hyperplane general} with the $\eps$ therein set to $\frac{1}{8s}$ (which indeed satisfies $\frac{1}{8s} \leq \frac{\delta_0}{6}$). The conclusion suffices for our proof.
It therefore remains to verify that $M$ is large enough to satisfy the premise of \cref{lm: hyperplane general}.
In particular, we check that 
\begin{equation}
    \label{eq:RM:verify large M}
M \overset{?}{\geq} \max\left\{\frac{10^5 q^4}{1/(8s)}, \frac{20q^{t_{q,d}}}{c\left(\frac{1}{8s}\right)} \right\},
\end{equation}
where $c( \cdot )$ is the soundness function of \cref{def: t test assump} for Reed-Muller codes. 

Clearly, we have $M \geq \frac{10^5 q^4}{1/(8s)}$, which leaves the need to verify $M$ is larger than the second term. By \Cref{thm: rm local testability}, the $t_{q,d}$-flat test has soundness function $c(\eps) = q^{-\KMconst} \cdot \min\set{1, q^{t_{q,d}}\eps}$ for Reed-Muller codes, and in particular $c\left(\frac{1}{8s}\right) = q^{-\KMconst} \cdot \frac{q^{t_{q,d}}}{8s}$, since $8s \geq q^{t_{q,d}}$. This implies 
\[
    \frac{20q^{t_{q,d}}}{c\left(\frac{1}{8s}\right)}
    = 20q^{t_{q,d}} \cdot q^{\KMconst}\cdot\frac{8s}{q^{t_{q,d}}} = 160q^{\KMconst} \cdot s \leq M \ ,
\]
which shows our choice of $M$ indeed satisfies \eqref{eq:RM:verify large M}, concluding the proof.
\end{proof}

\subsubsection{Small distance.}

\begin{lemma}\label{lm: small distance general}
    If $\eps_f\leq\frac{1}{2s}$ then 
    $$\Pr\big[\mathcal{T}_k^f \text{ rejects}\big]\ge \frac{s\eps_f}{8}\geq \frac{s\eps}{8}.$$
\end{lemma}
\begin{proof}
Let $g^* \in \RM[n,q,d]$ be such that $\dist(f,g)=\eps_f$. Let $S_1$ be the first $1/10$ of the points chosen in $S$ and $S_2$ be the remaining $9/10$ of the points chosen in $S$, so $|S_2| = 9s/10$.

Let $\Delta = \{x \in \F^n \; | \; f(x) \neq g^*(x) \}$. The tester proceeds by choosing a linear subspace $U \subset \F^n$ of dimension $k$ and then sampling $S_1, S_2 \subseteq U$.

For our analysis, it will be convenient to think of $S$ as being chosen as follows.

\begin{enumerate}
    \item Fix $U \subseteq \F^n$ to be a generic $k$-dimensional subspace of dimension $U$ and choose $S'_1 \subseteq U$ of size $|S_1|$.
    \item Choose $S'_2 \subseteq U$ of size $|S_2|$.
    \item Choose a random, full-rank linear transformation $T: \F^n \to \F^n$.
    \item Set $S_i = \{T(x) \; | \; x \in S_i\}$ for $i = 1,2$.
    \item Set $S = S_1 \cup S_2$.
\end{enumerate}
In the above view of $\mc{T}^f_k$, the $k$-dimensional subspace fixed by the tester is $T \circ U$, and the points queried are $S \subseteq T \circ U$. To analyze the rejection probability, we define the following events:
\begin{itemize}
    \item $E_1:$ the event that every pair of degree $d$ functions on $U$ (these correspond to $\RM[k,q,d]$) disagree on at least one point in $S'_1$. 
    \item $E_2:$ the event that $f$ and $g^*$ agree on $S_1$, i.e.\ $f|_{S_1} \equiv g^*|_{S_1}$.
    \item $E_3:$ the event that $S_2 \cap \Delta \neq \emptyset$.
\end{itemize}

One can check that if all three events above occur, then the tester rejects. Indeed, if $E_1$ and $E_2$ both occur, then $g^*|_{T \circ U} \in \RM[k,q,d]$ is the unique degree $d$ function over $T \circ U$ that agrees with $f$ on $S_1$. The tester will then reject if $g^*$ and $f$ disagree on any point on $S_2$ -- in other words, if $E_3$ occurs. Thus, 
\[
\Pr[\mc{T}^f_k \text{ rejects}] \geq \Pr_{S_1, S_2}[E_1 \land E_2 \land E_3] =  \Pr_{S'_1}[E_1] \cdot \Pr_{S'_1, S'_2, T}[E_2 \land E_3 \; | \; E_1]. 
\]
We now bound each term appearing above. First note that $E_1$ does not depend on the identity of the subspace $U$. By \cref{claim:sample_based_ub general}, we have 
\[
\Pr_{S'_1}[E_1] \geq 1- e^{-\delta_0 |S_1|} |\RM[k,q,d]|^2 \geq \frac{1}{2}.
\]
For the second term, we have
\[
 \Pr_{S'_1, S'_2, T}[E_2 \land E_3 \; | \; E_1] \geq  \Pr_{S'_1, S'_2, T}[ E_3 \; | \; E_1] -  \Pr_{S'_1, S'_2, T}[\overline{E_2} \; | \; E_1].
\]
We will bound each term on the right-hand side. Let us suppose $S'_1$ and $S'_2$ have already been sampled, and now bound the probabilities over random $T$.

For the first term, we use the principle of inclusion-exclusion to get a lower bound:
\begin{align*}
 \Pr_{T}[ E_3 \; | \; E_1] &\geq \sum_{x \in S'_2 \setminus\{0\}} \Pr_{T}[T(x) \in \Delta] - \sum_{x, y \in S'_2\setminus\{0\}, x\neq y,}\Pr_{T}[T(x), T(y) \in \Delta] \\
 &\geq \frac{9s/10-1}{2}\eps_f - 2 \frac{(9s/10-1)^2}{4}\eps_f^2 \\
 &\geq \frac{3s\eps_f}{8}\ .
\end{align*}
In the first transition we are using the fact that $T(x)$ is uniformly random over random $T$ and the fact that $T(x), T(y)$ are nearly independent for $x \neq y$. In the last transition we are using the fact that $s\eps_f/4 \leq 1/8$ by assumption. 

To upper bound the second term (again as a probability over $T$, fixing $S'_1, S'_2$), we use a union bound:
\[
\Pr_{S'_1, S'_2, T}[\overline{E_2} \; | \; E_1] \leq \sum_{x \in S'_1}\Pr_{T}[T(x) \in \Delta] \leq \frac{s\eps_f}{10}.
\]

Putting everything together, we get
\[
\Pr[\mc{T}^f_k \text{ rejects}] \geq \Pr_{S'_1}[E_1] \cdot \Pr_{S'_1, S'_2, T}[E_2 \land E_3 \; | \; E_1] \geq \frac{1}{2} \cdot \frac{s\eps_f}{4} \geq \frac{s\eps_f}{8}\ .
\]
\end{proof}

\subsubsection{Large distance.}
We have already shown that the tester has the desired soundness when $\eps_f \leq  \frac{1}{2s}$. To complete the proof of \cref{th: semi-sample rm}, it remains to show the desired soundness when $\eps_f > \frac{1}{2s}$. In particular, the following lemma suffices.
\begin{lemma} \label{lm: large distance rm}
    Set $\eta := \frac{1}{128}$ and $\gamma := 10^7 q^{\max\set{4, \KMconst}}\cdot s$.
    The following holds for every $n \geq k$. Let $f\colon\F_q^n\to\F_q$ be such that $\dist(f, \RM[n, q, d]) = \eps_f\geq 16\eta/s$. Then 
    $$\Pr\big[\mathcal{T}_k^f \text{ rejects}\big]\geq \eta+\frac{\gamma}{q^n}.$$
\end{lemma}

We will analyze the large distance case by induction of $n$. Consider the following tester $\mathcal{T}'_k$: the tester first picks a (linear) hyperplane $W$ and then
simulates $\mathcal{T}_k^{f|_{W}}$. That is, the tester chooses a random $k$-dimensional subspace $U \subseteq W$ and then performs the sample-based tester inside of $U$. We note that  $U$ is still a uniformly random $k$-dimensional subspace in $\F^n$, so it follows that
\begin{equation} \label{eq: decomp rej}
\Pr[\mathcal{T}_k^f  \text{ rejects}]= \E_{W}\left[\Pr[\mathcal{T}_k^{f|_W}\text{ rejects}]\right].
\end{equation}

Therefore, in order to conclude that the rejection probability is high, we will invoke \cref{lm: hyperplane deg d} to say that for nearly all of the hyperplanes $W$, the distance of $f|_W$ to $\mc{C}_{n-1}$ is still within a constant factor of $\eps_f$. As $f|_W$ and $\mc{C}_{n-1}$ are one dimension down, we can conclude that $\mathcal{T}_k^{f|_W}$ rejects with good probability. 
Since \cref{lm: hyperplane deg d} allows us to conclude $\dist(f|_W, \mc{C}_{n-1})$ can only increase by more than a constant factor for a number of $W$ that is independent of the dimension $n$, the above argument is sufficient for the inductive step.

\begin{proof}[Proof of \cref{lm: large distance rm}]
    The proof is by induction on the ambient dimension, $n$. In the base case, $n=k$, and $\mc{T}^f_k$ is simply the sample-based tester of \Cref{lm: sample-based lifted codes soundness}.
    Therefore, $\mc{T}^f_k$ rejects with probability at least $\min\{1/4,s\eps_f/8\} \ge 2\eta$, using the premise and the definition of $\eta$.
    Next, apply \Cref{fact: sk small} with $c = \max\set{4,\KMconst} + 31$, noting that by \eqref{eq:RM:choice of k}, we have large enough $k$, yielding 
    \[
        \frac{\gamma}{q^k} 
        = \frac{10^7 q^{\max\set{4, \KMconst}}\cdot s_k}{q^k}
        \leq 10^7 q^{\max\set{4,\KMconst}} q^{-c} < \frac{1}{128} = \eta ,
    \]
    which implies rejection probability of at least $2\eta \geq \eta + \gamma/q^n$, concluding the base case.

   For the induction step, suppose that the lemma holds for testing all functions over $\F_q^{n'}$  that are sufficiently far from $\RM[n', q, d]$, for all $k \leq n' < n$.
   We will show that the lemma also holds for $f: \F_q^n \to \F_q$ with $\eps_f \geq 16\eta / s$.
    
    Let $\mathcal{H}$ denote the set of all hyperplanes in $\F^n$ and let $N=|\mathcal{H}|=(q^n-1)/(q-1)$. Define 
    \[
    \mathcal{H}^*=\{W \in \mathcal{H}\ :\ \dist(f|_W,\mc{C}_{n-1})<16\eta/s\} \ ,
    \]
    and let $K=|\mathcal{H}^*|$. Using \eqref{eq: decomp rej} and the induction hypothesis, we can decompose the rejection probability as follows: 
    \begin{align*}
        \Pr[\mathcal{T}_k^f  \text{ rejects}] &= \E_{W \in \mc{H}}\left[\Pr[\mathcal{T}_k^{f|_W}\text{ rejects}]\right] \\
        &\geq \E_{W \in \mc{H} \setminus \mc{H}^*}\left[\Pr[\mathcal{T}_k^{f|_W}\text{ rejects}]\right] - \Pr_{W \in \mc{H}}[W \in \mc{H}^*] \\
        &\geq \eta + \frac{\gamma}{q^{n-1}} - \frac{K}{N},
    \end{align*}
    where we are applying the induction hypothesis on $f|_W$ in the third transition.

    \vspace{0.2cm}

    \noindent \textbf{Case 1: $K \leq \gamma \cdot \frac{q^n - 1}{q^n}$.} 
    
    \vspace{0.1cm}
    
    In this case, $\frac{K}{N} \leq \gamma \cdot \frac{q-1}{q^n}$. Then,
    \[
    \Pr[\mathcal{T}_k^f  \text{ rejects}] \geq \eta + \frac{\gamma}{q^{n-1}} - \frac{K}{N} \geq \eta +\frac{\gamma}{q^n},
    \]
    and we are done.
    
    \vspace{0.2cm}
    
    \noindent \textbf{Case 2: $K> \gamma \cdot \frac{q^n - 1}{q^n}$.} 
    
    \vspace{0.1cm}

   In this case $K > \gamma/2 > 10^6 q^{\max\set{4,\KMconst}} \cdot s$, which means there are at least this many hyperplanes $W\in\mathcal{H}$ on which $f|_W$ is $\frac{1}{8s}$-close to degree $d$ (by the definition of $\mathcal{H}^*$ and $\eta = 1/128$).
   This allows us to invoke \Cref{lm: hyperplane deg d} and deduce that $\eps_f \leq 1/(2s)$. 
    Finally, applying \cref{lm: small distance general} shows that 
    \[
        \Pr[\mathcal{T}_k^f  \text{ rejects}]
        \geq \frac{s\eps_f}{8}
        \geq \eta +\frac{\gamma}{q^n},
    \]
    where the second inequality uses the same argument as in the induction base case.
\end{proof}

\subsection{Lifted Affine Invariant Codes: Proof of \texorpdfstring{\cref{th: semi-sample lifted}}{Semi-sample-based Tester Theorem }} \label{sec: lifted test proof}
The proof of \cref{th: semi-sample lifted} is identical to that of \cref{th: semi-sample rm} for Reed-Muller codes, except we apply \cref{lm: hyperplane general} with the appropriate parameters for lifted affine-invariant codes instead. We also require a specialization of \cref{lm: hyperplane general} for $\mc{C}_n$, which takes the place of \cref{lm: hyperplane deg d} as well as a bound on $Q_k/q^k$ similar to \cref{fact: sk small}. We state these two lemmas below. Recall that in this subsection, $\delta_0$ is the smallest minimum distance of the codes $\mc{C}_t, \ldots, \mc{C}_n$.

\begin{lemma} \label{lm: q_k small}
    \[
    \frac{Q_k}{q^k} \leq \frac{1}{500}.
    \]
\end{lemma}
\begin{proof}
    This follows immediately from the assumption 
    \[
    \frac{\log_q |\mc{C}_k|}{q^k} \leq \frac{\delta_0}{10^6 \log(q)}
    \]
    in \cref{th: semi-sample lifted}.
\end{proof}
\begin{lemma} \label{lm: hyperplane affine lifted}
Suppose there are $M =q^{\Theta(1)} \cdot Q_k$ hyperplanes $W_1, \ldots, W_M \subset \F^n$ such that for each $i = 1, \ldots, M$, $f|_{W_i}$ is $\frac{1}{8Q_k}$-close to a function in $\mc{C}_{n-1}$. Then $f$ is $\frac{1}{2Q_k}$-close to $\mc{C}_n$.
\end{lemma}
\begin{proof}
We apply \cref{lm: hyperplane general} for $\mc{C}_n$ with the $\eps$ therein set to $\frac{1}{8Q_k}$. Note that we can do so because $\frac{1}{8Q_k} \leq \frac{\delta_0}{6}$. It remains to check that $M$ is large enough to apply \cref{lm: hyperplane general}. In particular, we check that 
\[
M \geq \max\left\{\frac{10^5 q^4}{1/(8Q_k)}, \frac{20q^{t}}{c\left(\frac{1}{8Q_k}\right)} \right\},
\]
where $c( \cdot )$ is the soundness function of \cref{def: t test assump} for the lifted affine-invariant code, $\mc{C}_n$. It is clear that $M \geq \frac{10^5 q^4}{1/(8Q_k)}$, so we need only check it is also at least the second value on the right-hand side above

\cref{thm: testing lifted km} says that if $f$ is $\frac{1}{8Q_k}$ far from $\mc{C}_n$, then the $t$-space test rejects with probability at least, so 
\[
c\left(\frac{1}{8Q_k} \right) \geq q^{-O(1)} \min\left\{1, \frac{q^{t}}{8Q_k}\right\} > \frac{20q^{t}}{M},
\]
and this completes the proof.
\end{proof}

Now that we have \cref{lm: q_k small} and \cref{lm: hyperplane affine lifted}, the proof of \cref{th: semi-sample lifted} goes through by following \cref{th: semi-sample rm} and using \cref{lm: q_k small} and \cref{lm: hyperplane affine lifted} in place of \cref{fact: sk small} and \cref{lm: hyperplane deg d} respectively.

\begin{proof}[Proof of \cref{th: semi-sample lifted}]
    Fix a dimension parameter $k \geq 20t + 100$, and henceforth write $\mc{T}^f_k$ in place of $\mc{T}^f_k(Q_k, \mc{C}_n)$ for simplicity. Also let $\eps_f := \dist(f, \mc{C}_n)$ and assume $\eps_f \geq \eps$. We similarly split the analysis into a small distance and large distance case. 

    \paragraph{Small Distance Case:} Here, we similarly assume that $\eps_f \leq \frac{1}{2Q_k}$. Then the exact same proof as \cref{lm: small distance general} goes through, except we perform a union bound over $\mc{C}_k$ instead. Specifically, after $\mc{T}^f_k$ fixes a $k$-dimensional subspace $U$, we think of the $Q_k$ queries as being split into two sets, $S_1 = T \circ S'_1$ and $S_2  = T \circ S'_2$ for uniformly random $S'_1, S'_2$ of sizes $Q_k/10, 9Q_k/10$ respectively. Let $g^* \in \mc{C}_n$ be the closest function in $\mc{C}_n$ to $f$. We then bound the probability that the following events hold:
    \begin{itemize}
        \item $E_1$: the event that every pair of functions from $\mc{C}_k$ disagree on at least one point in $S'_1$. 
        \item $E_2$: the event that $f$ and $g^*$ agree on $S_1$.
        \item $E_3$: the event that $f$ and $g^*$ disagree on at least one point in $S_2$.
    \end{itemize}
    
    The number of queries, $Q_k$ is adjusted relative to $|\mc{C}_k|$ and $\delta_0$, so that we can again get
    \[
    \Pr[\mc{T}^f_k \; rejects] \geq  \Pr[E_1 \land E_2 \land E_3] \geq \frac{Q_k \eps_f}{8} \geq \frac{Q_k \eps}{8},
    \]
    using the exact same strategy in the small distance case.

    \paragraph{Large Distance Case:} Here we similarly assume $\eps_f > \frac{1}{2Q_k}$. The proof follows from a similar induction. All of the steps are exactly the same as in the proof of \cref{lm: large distance rm} except we apply \cref{lm: q_k small} and \cref{lm: hyperplane affine lifted} in place of \cref{fact: sk small} and \cref{lm: hyperplane deg d} respectively. Specifically, we change $\gamma$ in the large distance case to $q^{\Theta(1)}s$ now and keep $\eta$ the same. Then in the base case, we use \cref{lm: q_k small} to bound $\frac{\gamma}{q^k} \leq \frac{1}{128}$ and establish the base case. 

    For the inductive step we have the same Case $1$ and Case $2$ (except with the new value of $\gamma$). Case 1 is analyzed in the same way, while Case 2 uses \cref{lm: hyperplane affine lifted}.
\end{proof}

\subsection{All Affine Invariant Families: Proof of \texorpdfstring{\cref{thm:semi-sample-based soundness general}}{Semi-sample-based General Tester Theorem}} \label{sec: general test proof}

The approach of using induction as in \cite{BKSSZ10} combined with our hyperplane lemma (\cref{lm: hyperplane general} is fairly versatile and allows us to analyze the semi-sample-based tester for any affine invariant family of functions, even non-linear ones. Once again, the completeness is clear and we focus on the proof of soundness. Let $\eps_f := \dist(f, \mc{C}_n)$ and suppose $\eps_f \geq \eps$. The soundness analysis is again essentially the same as that of the Reed-Muller case, except some care is needed to make the parameters above work. 

We set 
\[
\gamma = 2 \cdot \max\left\{\frac{20q^r}{c\left(\frac{1}{8s_k} \right)}, q^{\Theta(1)} \cdot s_k \right\}.
\]
To start, we observe that $\gamma$ is small relative to $q^k$.
\begin{lemma}
 For the value of $k$ chosen, we have,
\[
\frac{q^k}{1000} \geq \gamma.
\]
\end{lemma}
\begin{proof}
    This follows immediately from the assumptions about $k$ in \cref{thm:semi-sample-based soundness general}
\end{proof}

\paragraph{Small Distance Case:} 
The small distance case is exactly the same as in that of \cref{th: semi-sample rm}. Its proof is also the same so we omit it.
\begin{lemma}\label{lm: small distance non-lin}
    If $\eps_f\leq\frac{1}{2s_k}$ then 
    $$\Pr\big[\mathcal{T}_k^f \text{rejects}\big]\ge \frac{s_k\eps_f}{8}\geq \frac{s_k\eps}{8} \geq \eta.$$
\end{lemma}

\paragraph{Large Distance Case:}
\begin{lemma} \label{lm: large distance non-lin}
    The following holds for every $n \geq k$. Let $f\colon\F_q^n\to\F_q$ be such that $\dist(f, \mc{C}_n) = \eps_f\geq \frac{1}{8s_k}$. Then 
    $$\Pr\big[\mathcal{T}_k^f \text{ rejects}\big]\geq \eta+\frac{\gamma}{q^{k}}.$$
\end{lemma}
\begin{proof}
The proof is by induction on the ambient dimension $n$. In the base case $n = k$ and $\mc{T}^f_k$ is the sample-based tester of \cref{lm: sample-based lifted codes soundness}. Then, we get that $\mc{T}^f_k$ rejects with probability at least 
\[
\min \{1/4 , s_k \eps_f / 8 \} \geq \frac{1}{128} + \frac{\gamma}{q^k}.
\]
For $\eps_f \geq \frac{1}{8s_k}$, so the left hand side above is just $1/64$, while we choose $k \geq \Omega(\log_q(\gamma))$, so $ \frac{\gamma}{q^k} < 1/128$.

Towards the inductive step, let $\mathcal{H}$ denote the set of all hyperplanes in $\F^n$ and let $N=|\mathcal{H}|=(q^n-1)/(q-1)$. Define 
    \[
    \mathcal{H}^*=\{W \in \mathcal{H}\ :\ \eps_{f|_W}=\dist(f|_W,\mc{C}_{n-1})<\frac{1}{8s_k}\}.
    \]
    Let $K=|\mathcal{H}^*|$. Using \eqref{eq: decomp rej} and the induction hypothesis, we can decompose the rejection probability as follows: 
    \begin{align*}
        \Pr[\mathcal{T}_k^f  \text{ rejects}] &= \E_{W \in \mc{H}}\left[\Pr[\mathcal{T}_k^{f|_W}\text{ rejects}]\right] \\
        &\geq \E_{W \in \mc{H} \setminus \mc{H}^*}\left[\Pr[\mathcal{T}_k^{f|_W}\text{ rejects}]\right] - \Pr_{W \in \mc{H}}[W \in \mc{H}^*] \\
        &\geq \eta + \frac{\gamma}{q^{n-1}} - \frac{K}{N},
    \end{align*}
    where we are applying the induction hypothesis on $f|_W$ in the third transition.

    \vspace{0.2cm}

    \noindent \textbf{Case 1: $K \leq \gamma \cdot \frac{q^n - 1}{q^n}$.} 
    
    \vspace{0.1cm}
    
    In this case, $\frac{K}{N} \leq \gamma \cdot \frac{q-1}{q^n}$. Then,
    \[
    \Pr[\mathcal{T}_k^f  \text{ rejects}] \geq \eta + \frac{\gamma}{q^{n-1}} - \frac{K}{N} \geq \eta +\frac{\gamma}{q^n},
    \]
    and we are done.
    
    \vspace{0.2cm}
    
    \noindent \textbf{Case 2: $K> \gamma \cdot \frac{q^n - 1}{q^n}$.} 
    
    \vspace{0.1cm}

    We apply \Cref{lm: hyperplane general} with the quantity $\eps$ in \Cref{lm: hyperplane general} set to $\frac{1}{8s_k}$. This application yields
    \[
    \eps_f\leq 4 \cdot \frac{1}{8s_k} = \frac{1}{2s_k}.
    \]
   Note that $\gamma$ and $K$ were specifically chosen to be large enough to apply \cref{lm: hyperplane general} as
    \[
    K > \gamma \cdot \frac{q^n - 1}{q^n} \geq \max\left\{\frac{20q^t}{c(16\eta/s_k)}, \frac{10^5q}{(16\eta/s_k)} \right\}.
    \]
    Since $ \eps_f \leq \frac{64\eta}{s_k} \leq \frac{1}{2s_k}$, we have, by, \cref{lm: small distance non-lin}, that the rejection probability is at least 
    \[
    \frac{s_k\eps_f}{8} \geq \frac{1}{64} \geq \eta+\gamma/q^n.
    \] 
   For the last inequality, we are using
    \[
    \frac{\gamma}{q^n} \leq \frac{\gamma}{q^{k}} \leq \frac{1}{64} - \eta.
    \]
\end{proof}
\section{Application to Testing with an Online-Adversary}
We start by formally introducing the online-adversary model and discuss the testers we obtain for this model afterwards. Specifically we show that \cref{th: semi-sample rm} and \cref{th: semi-sample lifted} lead to testers for Reed-Muller codes and lifted affine-invariant codes respectively in the online-adversary model. For other affine invariant families, \cref{thm:semi-sample-based soundness general} could also give a tester in the online adversary model, but we are not aware of any families where the currently known testers have strong enough soundness for our theorem to apply.

\subsection{The Online-Manipulation Model}
The online-manipulation-resilient testing model was originally proposed by Kalemaj, Raskhodnikova, and Varma \cite{KalemajRV22}, and with extended variants was later formally defined by Ben-Eliezer, Kelman, Meir, and Raskhodnikova \cite{ben2024property}.
We follow the definitions of the latter, as specified below.

The input is accessed via a sequence $\set{\oracle{i}}_{i\in {\mathbb N}}$ of oracles, where $\oracle{i}$  is used to answer the $i$-th query. The oracle $\oracle{1}$ gives access to the original input (e.g., when the input is a function $f$, we have $\oracle{1} \equiv f$), and subsequent oracles are objects of the same type as the input (e.g., functions with the same domain and range). Each such oracle is obtained by the adversary by modifying the previous oracle to include a growing number of erasures/corruptions as $i$ increases.
We use $\Dist(\oracle{},\oracle{}')$ for the Hamming distance between the two oracles (i.e., the number of queries for which they give different answers).
We let $t \in \N$  denote the number of \emph{erasures} (or \emph{corruptions}) \emph{per query}. 

\begin{definition}[Fixed-rate and budget-managing adversaries]\label{def:fixed-rate_budget-managing}
     Fix a parameter $t>0$. A sequence\footnote{Our algorithms only access a finite subsequence of this sequence.} of oracles $\oracle{} =  \set{\oracle{i}}_{i\in {\mathbb N}}$
     is induced by a \emph{$t$-online fixed-rate adversary} if $\oracle{1}$ is equal to the input and, for all $i\in \N$, 
     \[
        \Dist(\oracle{i}, \oracle{i+1}) \leq \floor{(i+1)\cdot t} - \floor{i\cdot t}.
     \]
    A sequence of oracles $\oracle{} =  \set{\oracle{i}}_{i\in {\mathbb N}}$
     is induced by a \emph{$t$-online budget-managing adversary} if $\oracle{1}$ is equal to the input and, for all $i \in \N$,
     \[
        \Dist(\oracle{1}, \oracle{i+1}) \leq i\cdot t.
     \]
\end{definition}

  Finally, we consider two types of manipulations to the input: erasures and corruptions.
\begin{definition}[Erasure and corruption adversaries]
Let $\perp$ represent the erasure symbol. A sequence of oracles $\oracle{} =  \set{\oracle{i}}_{i\in {\mathbb N}}$ is induced by an \emph{erasure adversary} if for all $i\in\N$ and data points~$x$,
$$\oracle{i+1}(x) \in \set{\oracle{i}(x), \perp} .$$
In contrast, a \emph{corruption adversary} can change answers to anything in the range, i.e., $\oracle{i}$ can be any valid input for the computational task at hand.
\end{definition}

A property $\mc{P}$ denotes a set of objects (typically, a set of functions). Intuitively, it represents the set of positive instances for the testing problem. The (relative Hamming) distance between a function $f$ and a property $\mc{P}$, denoted $dist(f,\mc{P})$, is the smallest fraction of function values of $f$ that must be changed to obtain a function in $\mc{P}.$ Given a proximity parameter $\eps\in(0,1)$, we say that $f$ is \emph{$\eps$-far} from $\mc{P}$ if $dist(f,\mc{P})\geq\eps.$
 An online tester is given a proximity parameter $\eps$ and the rate of erasures (or corruptions) $t$.
\begin{definition}[Online $\eps$-tester]\label{def:online_tester} Fix $\eps\in(0,1).$
    An online $\eps$-tester $\mc{T}$ for a property $\mathcal{P}$ that works in the presence of a specified adversary (e.g., $t$-online erasure budget-managing adversary) is given access to an input function $f$ via a sequence of oracles 
    $\oracle{} =  \set{\oracle{i}}_{i\in {\mathbb N}}$
     induced by that type of adversary.
    For all adversarial strategies of the specified type,
    \begin{enumerate}
        \item if $f \in \mathcal{P}$, then $\mc{T}$ accepts with probability at least 2/3, and
        
        \item if $f$ is $\eps$-far from $\mc{P},$
        then $\mc{T}$ rejects with probability at least 2/3, 
    \end{enumerate}
    where the probability is taken over the random coins of the tester. If $\mc{T}$ works in the presence of an erasure (resp., corruption) adversary, we refer to it as an online-erasure-resilient (resp., online-corruption-resilient) tester.
    
    If $\mc{T}$ always accepts all functions $f\in\mc{P}$, then it has \emph{1-sided error.} If $\mathcal{T}$ chooses its queries in advance, before observing any outputs from the oracle, then it is \emph{nonadaptive}.
\end{definition}

To ease notation, we use $\oracle{}(x)$ for the oracle's answer to query  $x$  (omitting the timestamp $i$). 
If $x$ was queried multiple times, $\oracle{}(x)$ denotes the first answer given by the oracle.

\subsection{Optimal Online-Manipulation-Resilient Tester for Reed-Muller}

By setting the dimension of the semi-sample-based tester large enough relative to the erasure rate $t$, we obtain testers for Reed-Muller codes in the $t$-online erasure model.
\begin{theorem} \label{thm: rm online erasure test}
Let $\eps$ be a distance parameter, for any erasure rate $t \leq \eps^2q^{n-O(d)}$  and set 
\[\ka=O(d)+\log_q(t/\eps^2)\] 
sufficiently large. Then given an input function $f: \F^n \to \F$, there is a tester for $\RM[n,q,d]$ with the following parameters in the $t$-online erasure model satisfying:
    \begin{itemize}
        \item Completeness: If $f \in \RM[n,q,d]$, the tester always accepts.
        \item Soundness: If $f$ is $\eps$-far from $\RM[n,q,d]$, then 
        \[
        \Pr[\mc{T}^f_k \text{ rejects}] \geq \frac{2}{3}
        \]
        \item Query Complexity:
        \[
        Q=\max\left\{\frac{1}{\eps}, 128s_{\ka} \right\} =  O\left(\frac{1}{\eps}+q^{2d}\left(1+O \left(\frac{\log_q(t/\eps)}{d}\right) \right)^d \right) 
        \]
    \end{itemize}
\end{theorem}
\begin{proof}
The tester is the semi-sample-based tester $\mc{T}_k^f(Q, \RM[n,q,d])$ from \cref{sec: semi-sample-based tester} with $k=\ka$. For the online erasure model, we also make the following slight modification: if any of the queried points is erased (i.e., answered with $\perp$), the tester will stop and accept. This modification maintains completeness while degrading the soundness by an amount that we will see is negligible. 

By \cref{th: semi-sample rm}, we can repeat the tester $O\left(\frac{1}{s_k\eps}+1\right)$ times and reject with probability at least $9/10$ if no erasure is ever queried, so it remains for us to upper bound the probability of querying an erasure. The total number of queries made is $Q=O(1/\eps+s_k)$ queries. 

It remains to bound the probability that the tester queries a point that is erased. We show that this probability is at most $1/5$ completing the proof. Indeed, since the erasure rate is $t$, the total number of erasure made during the run of the algorithm is at most $tQ$. Since every query is uniform within a $k$-space. the probability for each query to be erased is at most $tQ/q^k$, by a union bound, the probability that any of the $Q$ queries is erased is at most 
\begin{equation}\label{eq: see erasure prob}
    tQ^2/q^k=O\left(\frac{t}{q^k}\left(\frac{1}{\eps^2}+s_k^2\right)\right).
\end{equation} 
Note that for $k>d$ we have $s_k\le q^{2d}(10k/d)^{d}$. For large enough $\ka=O(d+\log_q(t/\eps))$ we have that this quantity is at most,
$$O\left(\frac{t}{q^k}\left(\frac{1}{\eps^2}+s_k^2\right)\right)\le O\left(\frac{t}{q^k\eps^2}
\right)+O\left(\frac{t}{q^k}\left(\frac{10qk}{d}\right)^{4d}\right)<1/5.$$
It follows that an $\eps$-far function is accepted with probability at most $1/10+1/5\le 1/3$, as required. 
\end{proof}

Using an observation by \cite{KalemajRV22}, we get a similar theorem for the online corruption model, wherein the adversary can change the value of $f$ at a point rather than just erase the value at that point. Now, with a two-sided error, an error may occur whenever the tester queries a manipulated point in either case.   
\begin{theorem}
Let $\eps$-be a distance parameter and set $\ka=O(d+\log_q(t/\eps^2))$ be large enough.   
 Then given an input function $f: \F^n \to \F$, there is a tester for $\RM[n,q,d]$ with the following parameters in the $t$-online \emph{corruption} model:
    \begin{itemize}
        \item Completeness: If $f \in \RM[n,q,d]$
        then 
        \[
        \Pr[\mc{T}^f_k \text{ accepts}] \geq \frac{2}{3}
        \]
        \item Soundness: If $f$ is $\eps$-far from $\RM[n,q,d]$, then 
        \[
        \Pr[\mc{T}^f_k \text{ rejects}] \geq \frac{2}{3}
        \]
        \item Query Complexity:
          \[
        \max\left\{\frac{1}{\eps}, 128s_{\ka} \right\} =  O\left(\frac{1}{\eps}+q^{2d}\left(1  + O\left(\frac{\log_q(t/\eps)}{d}\right)\right)^d \right) 
        \]
    \end{itemize}
\end{theorem}
\begin{proof}
For the completeness, note that $f$ is accepted as long as no corrupted point is queried. From the proof of \cref{thm: rm online erasure test}, this event happens with probability at most $1/5$, and thus the completeness is at least $2/3$. The analysis of the soundness is exactly the same as \cref{thm: rm online erasure test}.
\end{proof}

\subsection{A Matching Query Complexity Lower Bound}

\cite{ben2024property} showed a lower bound for the number of queries required by any erasure-resilient algorithm for testing Reed-Muller over $\F_2$. In their proof they relied on results by \cite{KeevashS05,Ben-EliezerHL12} for lower bounding the dimension of the punctured Reed-Muller code over $\F_2$. 
We note that their proof can easily extended to all fields using similar results by \cite{beameOY2018} and \cite{golovnevGHNY2024}.
We use the following lemma: 
\begin{lemma}[\protect{\cite[Corollary 1.4]{golovnevGHNY2024}}]
    Let $q$ be a prime power and  $S\subset\F_q^n$ such that $S=q^r$. Then the dimension of the subspace spanned by
  $$\{(f(x))_{f\in \RM[n,q,d]}|x\in S\}$$
  is at least $|M_q^n(q^r)_{\leq d}|$.
  Where $M_q^n(k)$ is the set of the first $k$ elements in $\{0,\dots ,q-1\}^n$ in lexicographic order, and $M_q^n(k)_{\leq d}$ are the elements $x\in M_q^n(k)$ such that $\sum x_i\leq d$.
\end{lemma}
We note that if $d\leq q$ then $|M_q^n(q^r)_{\leq d}|=\binom{r+d}{d}\geq \left(\frac{r}{d}\right)^d$.
Following is a generalized lower bound, we omit the proof since it is exactly the same as the proof of \cite[Theorem 1.2]{ben2024property}.
\begin{theorem}\label{thm:lower bound for all fields}
     Let $q$ be a prime power, fix an integer $d$. 
   There exists $n_0=n_0(d,q)$, such that for all $n\geq n_0$ and 
   $\eps \in(0, 1/3]$, every $\eps$-tester of $\RM[n,q,d]$ that works in the presence of a $t$-online erasure adversary must make $|M_q^n(t)_{\leq d}|$ queries.
   In particular, if $d\leq q$, we have the explicit lower bound of $|M_q^n(t)_{\leq d}|\geq\left(\frac{\floor{\log_q t}}{d}\right)^d$.
\end{theorem}

\subsection{Online-Manipulation-Resilient Testers for Lifted Affine-Invariant Codes}

By setting the dimension of the semi-sample-based tester to be sufficiently large compared to the erasure rate $t$, we obtain testers for the lifted codes $\mc{C}_n = \mc{C}^{r \nearrow n}$ in the $t$-online erasure model provided there exists some dimension $k$ such that 
\begin{equation} \label{eq: lifted t assump}  
\frac{Q_k^2 t}{q^k} \leq \frac{1}{100},
\end{equation}
where $\delta_0$ is the minimum distance of $\mc{C}_r, \ldots, \mc{C}_n$, and
\[
    Q_k = \ceil{\frac{100\ln |\mc{C}_k|}{\delta_{0}}}+ 1.
\]

\begin{theorem} \label{thm: lifted t online erasure tester}
Let $\mc{C}_n = \mc{C}^{r \nearrow n}$ be a lifted affine-invariant code, and let $t$ be an erasure parameter such that, for some dimension parameter $k$, \cref{eq: lifted t assump} is satisfied. Then for any distance parameter $\eps$ the code $\mc{C}_n$ has an $O(\eps^{-1} + Q_k)$ query tester in the $t$-online erasure model:
  \begin{itemize}
    \item Completeness: If $f \in \mc{C}_n$, then the tester accepts with probability $1$.
    \item Soundness: If $f$ is $\Omega(1)$-far from $\mc{C}_n$
    \[
    \Pr[\mc{T}^f_k \text{ rejects}] \geq \frac{2}{3}.
    \]
    \end{itemize}
\end{theorem}
\begin{proof}
The tester is obtained by repeating the tester $\mc{T}^f_k(Q_k, \mc{C}_k)$,  $Q_k + \eps^{-1}$ times. By \cref{th: semi-sample lifted}, if no erased point is queried, then $f$ is rejected with probability at least $9/10$. On the other hand, the probability of hitting an erasure is at most $Q_k^2 t /q^k$ which is at most $\frac{1}{100}$ by \eqref{eq: lifted t assump}, so overall $f$ which is $\eps$-far from $\mc{C}_n$ is rejected with probability at least $9/10 - 1/10 \geq 2/3$.
\end{proof}

Once again, the tester of \cref{thm: lifted t online erasure tester} also works in the online corruption model of \cite{KalemajRV22}.

\begin{theorem} \label{thm: lifted t online corruption tester}
Let $\mc{C}_n = \mc{C}^{r \nearrow n}$ be a lifted affine-invariant code and let $t$ be an erasure parameter such that for some dimension parameter $k$, \cref{eq: lifted t assump} is satisfied. Then for any distance parameter $\eps$ the code $\mc{C}_n$ has an $O(\eps^{-1} + Q_k)$ query tester in the $t$-online corruption model:
  \begin{itemize}
    \item Completeness: If $f \in \mc{C}_n$, then the tester accepts with probability $\frac{2}{3}$.
    \item Soundness: If $f$ is $\Omega(1)$-far from $\mc{C}_n$
    \[
    \Pr[\mc{T}^f_k \text{ rejects}] \geq \frac{2}{3}.
    \]
    \end{itemize}
\end{theorem}
\begin{proof}
    The completeness follows because a valid $f$ is accepted as long as no corrupted point is queried. By the proof of \cref{thm: lifted t online erasure tester}, this happens with probability at least $1/10$, and the completeness follows. The analysis of the soundness is the same as that of \cref{thm: lifted t online erasure tester}.
\end{proof}
\section{Further Directions}
Our work leaves a number of directions open for future work. We list these below.
\begin{itemize}
    \item Improved Soundness for \cite{KS08}: As mentioned in our introduction, there is a general result which shows any local characterization for an affine invariant family also yields a local tester \cite{KS08}. The soundness shown for these testers is sub-optimal however, so it would be interesting to see if our techniques could be useful towards improving the soundness of the general testers of \cite{KS08}.
    
    \item Testing Other Affine Invariant Families: It would be interesting to see if our techniques can yield new testing results for other affine invariant families of functions, particularly the nonlinear ones considered in \cite{bfhhl}. As mentioned, our analysis falls short in that the hyperplane agreement lemma currently relies on having an $r$-space tester with decent soundness, as opposed to the inexplicit soundness given in \cite{bfhhl}.
\end{itemize}

 \section*{Acknolegment} Part of this research was carried out during the authors’ visit to the Simons Institute for the Theory of Computing at UCB. 
 We wish to thank Pooya Hatami for bringing \cite{golovnevGHNY2024} and the references therein to our attention; Corollary 1.4 there essentially proves \autoref{thm:lower bound for all fields} by generalizing the lower from \cite[Theorem 1.2]{ben2024property} to all fields, showing that our new algorithm is optimal.

\bibliographystyle{alpha}
\bibliography{references}
\newpage
\appendix
\section{Weaker Version of \texorpdfstring{\cref{thm:semi-sample-based soundness general}}{the Semi-sample-based General Theorem} without Soundness Assumptions} \label{appendix: gen aff}
To state our result let us build up and recall some necessary notation. We let $\mathcal{C}_n \subseteq \{\F^n \to \F^n \}$ be an affine-invariant family of functions. Recall that for each dimension $k \leq n$, $\mc{C}_k = \{f|_U \; | \;  \dim(U) = k, f \in \mc{C}_n\}$. We assume that $\mc{C}_n$ is testable via the $r$-space test with soundness function $c(\cdot)$ and let $\delta_0$ be the smallest minimum distance of $\mc{C}_r, \ldots, \mc{C}_n$. 
\[
s_k=\ceil{\frac{10\ln |\mc{C}_n|}{\delta_{0}}}+ 1,
\]
for any dimension parameter $k$. Let $c_{\sem, k', k}$ denote the soundness function of $\mc{T}_k$ for testing $\mc{C}_{k'}$ where $k' \geq k$. That is, if $g: \F_q^{k'} \to \F_q$ is $\eps$-far from $\mc{C}_{k'}$,
\[
\Pr[\mc{T}^f_{k'} \; rejects] \geq c_{\sem, k', k}(\eps).
\]
Importantly, we note that we can obtain a lower bound on $c_{\sem, k', k}(\eps)$ that is independent of $k'$. Specifically,
\begin{equation}\label{eq: appendix}
c_{\sem, k', k}(\eps) \geq c(\eps) \cdot c_{\sem, k,k}(q^{-k}).
\end{equation}

Now, let us fix a dimension $k$ such that
\begin{equation} \label{eq: dim assumptions}
\frac{s_{k}}{q^{k}} \leq \frac{1}{100}.
\end{equation}
Note that 
\[
\frac{s_{k}}{q^{k}} = O\left(\frac{\log(q) \cdot R_{k}}{\delta_0}\right),
\]
where $R_{k'}$ is the rate of $\mc{C}_{k'}$, so for many affine invariant families, such a dimension $k$ exists. Now let
\[
\eta = c(\eps) \cdot c_{\sem, k, k}(q^{-k})
\]
\[
\gamma = \frac{20q^{r+\Theta(1)}}{c(16\eta/s)} \quad \text{and} \quad  k^\star = \Omega(\log(\gamma) + \log(1/\eta)).
\]

\begin{theorem}\label{thm:semi-sample-based soundness general weak}
Keep $\mc{C}_n$, $\delta_0$, $k$, $\gamma$, and $k^\star$ as above and let $f: \F^n \to \F$ be an input function and let $\eps$ be a distance parameter. Then, the tester $\mc{T}^f_k(s_k, \mc{C}_k)$ for the affine invariant family $\mc{C}_n$ has the following properties
    \begin{itemize}
        \item Completeness: If $f \in \mc{C}_n$, then the tester accepts with probability $1$.
        \item Soundness: If $f$ is $\eps$-far from $\mc{C}_n$ then 
        \[
        \Pr[\mc{T}^f_k ] \geq \eta 
        \]
    \end{itemize}
The query complexity of the tester is $s_{k}$, and we remark that $\eta$ is independent of the dimension of the domain, $n$.
\end{theorem}

Once again, the completeness is clear and we focus on the proof of soundness. Let $\eps_f := \dist(f, \mc{C}_n)$ and suppose $\eps_f \geq \eps$. The soundness analysis is essentially the same as that of the Reed-Muller case, except some care is needed to make the parameters above work. 
\subsection{Small Distance}
The small distance case is exactly the same as in that of \cref{th: semi-sample rm}. Its proof is also the same so we omit it.
\begin{lemma}\label{lm: small distance non-lin app}
    If $\eps_f\leq\frac{1}{2s}$ then 
    $$\Pr\big[\mathcal{T}_k^f \text{rejects}\big]\ge \frac{s\eps_f}{8}\geq \frac{s\eps}{8} \geq \eta.$$
\end{lemma}
\subsection{Large Distance}

\begin{lemma} \label{lm: large distance non-lin app}
    The following holds for every $n \geq k$. Let $f\colon\F_q^n\to\F_q$ be such that $\dist(f, \mc{C}_n) = \eps_f\geq \frac{1}{8s}$. Then 
    $$\Pr\big[\mathcal{T}_k^f \text{ rejects}\big]\geq \frac{\eta}{2}+\frac{\gamma}{q^{\max\{k, k^\star\}}}.$$
\end{lemma}
\begin{proof}
We first consider two cases based on the ambient dimension $n$.

\noindent \textbf{Case 1:} $n \leq k^\star$. It is sufficient to consider $n = k^\star$. Indeed otherwise, we can define the auxiliary function $f': \F^{k^\star} \to \F$ obtained by $f'(x,y) = f(x)$, where $x \in \F^k, y \in \F^{k^\star-k}$. In this case, the rejection probability is at least 
\[
c_{\sem, k^\star, k}(\eps) \geq \eta \geq \frac{\eta}{2} + \frac{\gamma}{q^{\max\{k, k^\star\}}}.
\]
by definition of $c_{\sem, k^\star, k}$ and \eqref{eq: appendix}.

\vspace{0.1cm}

\noindent \textbf{Case 2:} $n \geq k^\star$. The proof is by induction of the ambient dimension, $n$. Let us suppose that the lemma holds for testing all functions over $\F_q^{n'}$  for all $k^\star \leq n' < n$. We will show that the lemma also holds for $f: \F_q^n \to \F_q$ with $\eps_f \geq \frac{1}{8s}$.
    
    Let $\mathcal{H}$ denote the set of all hyperplanes in $\F^n$ and let $N=|\mathcal{H}|=(q^n-1)/(q-1)$. Define 
    \[
    \mathcal{H}^*=\{W \in \mathcal{H}\ :\ \eps_{f|_W}=\dist(f|_W,\mc{C}_{n-1})<\frac{1}{8s}\}.
    \]
    Let $K=|\mathcal{H}^*|$. Using \eqref{eq: decomp rej} and the induction hypothesis, we can decompose the rejection probability as follows: 
    \begin{align*}
        \Pr[\mathcal{T}_k^f  \text{ rejects}] &= \E_{W \in \mc{H}}\left[\Pr[\mathcal{T}_k^{f|_W}\text{ rejects}]\right] \\
        &\geq \E_{W \in \mc{H} \setminus \mc{H}^*}\left[\Pr[\mathcal{T}_k^{f|_W}\text{ rejects}]\right] - \Pr_{W \in \mc{H}}[W \in \mc{H}^*] \\
        &\geq \eta + \frac{\gamma}{q^{n-1}} - \frac{K}{N},
    \end{align*}
    where we are applying the induction hypothesis on $f|_W$ in the third transition.

    \vspace{0.2cm}

    \noindent \textbf{Case 1: $K \leq \gamma \cdot \frac{q^n - 1}{q^n}$.} 
    
    \vspace{0.1cm}
    
    In this case, $\frac{K}{N} \leq \gamma \cdot \frac{q-1}{q^n}$. Then,
    \[
    \Pr[\mathcal{T}_k^f  \text{ rejects}] \geq \eta + \frac{\gamma}{q^{n-1}} - \frac{K}{N} \geq \eta +\frac{\gamma}{q^n},
    \]
    and we are done.
    
    \vspace{0.2cm}
    
    \noindent \textbf{Case 2: $K> \gamma \cdot \frac{q^n - 1}{q^n}$.} 
    
    \vspace{0.1cm}

    We apply \Cref{lm: hyperplane general} with the quantity $\eps$ in \Cref{lm: hyperplane general} set to $\frac{1}{8s}$. This application yields
    \[
    \eps_f\leq 4 \cdot \frac{1}{8s} = \frac{1}{2s}.
    \]
   Note that $K$ is indeed large enough to apply \cref{lm: hyperplane general} as we set $\gamma$ so that
    \[
    K > \gamma \cdot \frac{q^n - 1}{q^n} \geq \max\left\{\frac{20q^t}{c(16\eta/s)}, \frac{10^5q}{(16\eta/s)} \right\}.
    \]
    Since $ \eps_f \leq \frac{64\eta}{s} \leq \frac{1}{2s}$, we have, by, \cref{lm: small distance non-lin}, that the rejection probability is at least 
    \[
    \frac{s\eps_f}{8} \geq \frac{1}{64} \geq \eta+\gamma/q^n.
    \] 
   For the last inequality, we are using
    \[
    \frac{\gamma}{q^n} \leq \frac{\gamma}{q^{k^\star}} \leq \frac{1}{64} - \eta.
    \]
\end{proof}
\end{document}